\documentclass[letterpaper,10pt,conference]{ieeeconf}

\IEEEoverridecommandlockouts	
\usepackage{graphicx}
\usepackage{epsfig} 

\usepackage{graphics}
\usepackage{subfigure}
\usepackage{stmaryrd}
\usepackage[normalem]{ulem}
\usepackage{amsmath,amssymb,url,listings,mathrsfs,wasysym}
\usepackage{paralist}
\usepackage{tikz}
\usetikzlibrary{arrows,automata,shapes,snakes}
\usepackage{algpseudocode}
\usepackage{algorithm}

\newcommand{\always}{\square}
\newcommand{\eventually}{\Diamond}
\renewcommand{\next}{\ocircle}
\newcommand{\until}{\hspace{1mm}\mathcal{U}\hspace{1mm}}

\newcommand{\true}{\relax\ifmmode \mathit{True} \else \em True \/\fi}
\newcommand{\false}{\relax\ifmmode \mathit{False} \else \em False \/\fi}
\newcommand{\aand}{\hspace{1mm}\wedge\hspace{1mm}}
\newcommand{\oor}{\hspace{1mm}\vee\hspace{1mm}}
\newcommand{\LTLX}{\text{LTL}_{\backslash\next}}

\newcommand{\DS}{\mathbb{D}}
\newcommand{\TS}{\mathbb{T}}
\newcommand{\lang}{\mathcal{L}_\omega}
\newcommand{\A}{\mathcal{A}}

\newcommand{\G}{\mathcal{G}}
\renewcommand{\P}{\mathcal{P}}
\newcommand{\PF}{\mathcal{PF}}
\newcommand{\R}{\mathcal{R}}
\newcommand{\SP}{\mathcal{SP}}
\newcommand{\ST}{\mathcal{ST}}
\newcommand{\X}{\mathcal{X}}
\newcommand{\Y}{\mathcal{Y}}

\newcommand{\word}{\sigma}
\newcommand{\AP}{\Pi}
\newcommand{\ap}{p}
\newcommand{\alphabet}{a}
\newcommand{\trace}[1]{\word_{#1}}

\newcommand{\cl}[1]{\overline{#1}}

\newcommand{\num}[1]{\relax\ifmmode \mathbb #1\else $\mathbb #1$\fi}
\newcommand{\naturals}{{\num N}}
\newcommand{\reals}{{\num R}}

\newtheorem{lem}{Lemma}
\newtheorem{prop}{Proposition}
\newtheorem{cor}{Corollary}
\newtheorem{defn}{Definition}

\newtheorem{exmp}{Example}

\urldef{\wongpiromsarn}\url{tichakorn@tcels.or.th}
\urldef{\topcu}\url{utopcu@seas.upenn.edu}
\urldef{\lamperski}\url{a.lamperski@eng.cam.ac.uk}
\title{\LARGE \bf Automata Theory Meets Barrier Certificates:\\ Temporal Logic Verification of Nonlinear Systems}
\author{Tichakorn Wongpiromsarn$^\star$%
 		\quad Ufuk Topcu$^\dagger$%
 		\quad Andrew Lamperski$^\ddagger$%
 		\thanks{$^\star$ Thailand Center of Excellence for Life Sciences, Thailand (\wongpiromsarn)}%
 		\thanks{$^\dagger$ , University of Pennsylvania, Philadelphia, PA (\topcu)}%
 		\thanks{$^\ddagger$ , University of Cambridge, United Kingdom (\lamperski)}}%

\begin{document}

\maketitle
\thispagestyle{empty}
\pagestyle{empty}

\begin{abstract}
We consider temporal logic verification of (possibly nonlinear) dynamical systems evolving over continuous state spaces.
Our approach combines automata-based verification and the use of so-called barrier certificates.
Automata-based verification allows the decomposition the verification task 
into a finite collection of simpler constraints over the continuous state space. 
The satisfaction of these constraints in turn can be (potentially conservatively) proved by appropriately constructed barrier certificates.
As a result, our approach, together with optimization-based search for barrier certificates, 
allows computational verification of dynamical systems against temporal logic properties while
avoiding explicit abstractions of the dynamics as commonly done in literature.
\end{abstract}

\section{Introduction}
We propose a sound but incomplete method for the computational verification of specifications expressed in temporal logic against the behavior of dynamical systems evolving over (potentially partially) continuous state spaces. This new method merges ideas from automata-based model checking with those from control theory including so-called barrier certificates and optimization-based search for such certificates. More specifically, we consider linear temporal logic (excluding the ``next" operator) formulas over atomic propositions that capture (sub)set memberships over the continuous state space. Under mild assumptions, the properties of the trajectories, which are salient for the verification, of the system can be characterized by infinite sequences (we call them {\it traces}) that track the atomic propositions satisfied along the corresponding trajectories (i.e., the subsets visited along the trajectory). Then, an automaton representation of the negation of the temporal logic formula guides a decomposition of the verification task into a finite collection of simpler constraints over the continuous state space. The satisfaction of these constraints in turn can be (potentially conservatively) proved by appropriately constructed barrier certificates.

Verification of dynamical systems against rich temporal logic specification has attracted considerable attention. A widely explored approach is based on proving (or disproving) (e.g., by using model checking \cite{Baier:PMC2008,clarkebook}) the specification using finite-state abstractions of the underlying dynamics \cite{TabuadaP03,asarin2007hybridization}. The consistency of the satisfaction of the specifications by the dynamical system and its finite-state abstractions is established through simulation and bi-simulation relations \cite{alur2000discrete} or approximately through approximate bi-simulation relations \cite{girard2007approximation}. In general, these existing approaches are not complete, except for certain simple dynamics \cite{HKPV98}. In addition, the abstract finite state systems are often large, leading to the state explosion problem.

The method we propose avoids explicit abstractions of the dynamics. On the other hand, the automaton representation of the specification may be interpreted as a ``minimal" finite-state abstraction required for verification. The details due to the dynamics ignored in this abstraction are then accounted for by the barrier certificates only to the level of fidelity and locally over the regions of the continuous state space dictated by the dynamics. However, similar to existing approaches for verifying nonlinear systems against temporal logic specifications, our approach is also not complete.

Not as rich as linear temporal logic but barrier certificates were originally considered to prove the satisfaction of temporal constraints, e.g., safety, reachability, and eventuality, for dynamical systems \cite{Vinter:1980,prajna05thesis}. Reference \cite{prajna05thesis} also demonstrated the use of multiple and/or more sophisticated\footnote{Informally in terms of the conditions that need to be satisfied by the corresponding barrier certificates.} barrier certificates for verifying properties beyond the basic ones mentioned above. Furthermore, one can imagine that it may be possible to look for increasingly complicated barrier certificates to verify arbitrary linear temporal logic specifications. The main contribution of this paper is to partly formalize such imagination by systematically constructing a collection of barrier certificates which all together witness the satisfaction of arbitrary linear temporal logic specifications.

The method developed in this paper is in principle applicable to a broad family of dynamical systems as long as certain, relatively mild smoothness conditions hold. In the presentation we consider continuous vector fields for simplicity. The step, which practically determines the applicability, of the proposed procedure is the computational search for barrier certificates. In this step, we focus on polynomial vector fields and resort to a combination of generalizations of the S-procedure \cite{parrilothesis:00,topcuthesis:08} and sum-of-squares relaxations for global polynomial optimization \cite{parrilothesis:00}. These techniques are relatively standard now in controls and have been used in other work on quantitative analysis of nonlinear and hybrid systems \cite{prajna05thesis,prajna2004nonlinear,jarvis2003some,simpaperTopcu,Tedrake:2010:LFM:1825484.1825491}.

The rest of the paper is organized as follows: We begin with some notation and preliminaries needed in the rest of the paper. The problem formulation in section \ref{sec:prob} is followed by the automata-theoretic notions in section \ref{sec:approach} which characterize the verification as checking properties of potentially infinitely many run fragments. Section \ref{sec:substrings} reduces this checking to a finite set of representative run fragments. Section \ref{sec:bc} discusses the role of the barrier certificates. Section \ref{sec:procedure} puts the pieces introduced in the earlier sections together and gives a pseudo-algorithm as well as pointers to some of the computational tools required to implement the algorithm. The critique in section \ref{sec:diss} is followed by an application of the method to an example, which is also used as a running example throughout the paper.

\section{Preliminaries}
\label{sec:prelim}
In this section, we define the formalism used in the paper to describe systems and their desired properties.
Given a set $X$, we let $2^X$ and $|X|$ denote the powerset and the cardinality of $X$, respectively, and
let $X^*$, $X^+$ and $X^\omega$  denote the set of finite, nonempty finite and infinite strings of $X$.
For finite strings $\word_1$ and $\word_2$, let 
$\word_1 \word_2$ denote a string obtained by concatenating $\word_1$ and $\word_2$,
$\word_1^*$ and $\word_1^+$ denote a finite string and a nonempty finite string, respectively, obtained by concatenating $\word_1$ finitely many times
and $\word_1^\omega$ denote an infinite string obtained by
concatenating $\word_1$ infinitely many times.
Given a finite string $\word = \alphabet_0 \alphabet_1 \ldots \alphabet_m$ where $m \in \naturals$
or an infinite string $\word = \alphabet_0 \alphabet_1 \ldots$,
a \emph{substring} of $\word$ is any finite string
$\alphabet_i \alphabet_{i+1} \ldots \alphabet_{i+k}$ where $i, k \geq 0$
and $i+k \leq m$ if $\word$ is finite.
%
Finally, for any $\Y \subseteq \reals^n$ where $n \in \naturals$, we let $\cl{\Y}$ be the closure of $\Y$ in $\reals^n$.

Consider a dynamical system $\DS$ whose state $x \in \X \subseteq \reals^n$, $n \in \naturals$
evolves according to the differential equation
\begin{equation}
\label{eq:sys}
\dot{x}(t) = f(x(t)).
\end{equation}
Let (by slight abuse of notation) $x: \reals_{\geq 0} \to \X$ also represent a trajectory of the system,
i.e., a solution of (\ref{eq:sys}).
We assume that the vector field $f$ is continuous to ensure that its solution
$x$ is piecewise continuously differentiable.

\subsection{Barrier Certificates}

We are interested in verifying the system in \eqref{eq:sys} against a broad class of properties (whose definition and semantics will be introduced later) that roughly speaking temporally and logically constrain the evolution of the system. A building block in the subsequent development is the use of the so-called barrier certificates which, in recent literature \cite{prajna05thesis}, were utilized to verify safety, reachability, and other simple specifications that can essentially be interpreted as instances of the specification language considered in this paper. We now introduce a barrier certificate-type result as a prelude. This result will later be invoked in section \ref{sec:bc}.

\begin{lem}
\label{lem:barrier-cert}
Let $\Y, \Y_0, \Y_1 \subseteq \X$.
Suppose there exists a differentiable function $B \ : \ \X \to \reals$ that satisfies the following conditions:
\begin{eqnarray}
  \label{eq: lem:barrier-cert1}
    &&B(x) \leq 0 \hspace{3mm}\forall x \in \Y_0,\\
  \label{eq: lem:barrier-cert2}
    &&B(x) > 0 \hspace{3mm}\forall x \in \cl{\Y_1},\\
  \label{eq: lem:barrier-cert3}
    &&\frac{\partial B}{\partial x}(x)f(x) \leq 0 \hspace{3mm}\forall x \in \cl{\Y} \setminus \cl{\Y_1}.
\end{eqnarray}
Then, any trajectory of $\DS$ that starts in $\Y_0$ cannot reach $\Y_1$ without leaving $\cl{\Y}$.
\end{lem}
\begin{proof}
Consider a trajectory $x$ of $\DS$ that starts in $\Y_0$.
Suppose $x$ reaches $\Y_1$ without leaving $\cl{\Y}$.
Then, there exists $T \in \reals$ such that $x(T) \in \cl{\Y_1}$ and
$x(t) \in \cl{\Y} \setminus \cl{\Y_1}$ for all $t \in [0, T)$.
From conditions (\ref{eq: lem:barrier-cert1}) and (\ref{eq: lem:barrier-cert2}), 
we get that $B(x(0)) \leq 0$ and $B(x(T)) > 0$.
In addition, condition (\ref{eq: lem:barrier-cert3}) implies that $B(x(t)) \leq 0$
for all $t \in [0, T)$.
From the continuity of $x$ and $B$, we can conclude that $B(x(T)) \leq 0$,
leading to a contradiction.
\end{proof}

Lemma \ref{lem:barrier-cert} (potentially conservatively) translates a verification question (whether all solutions to \eqref{eq:sys} satisfy the specified temporal ordering between ``visiting" $\mathcal Y_0,~\mathcal Y_1,$ and $\mathcal Y$) into search for a map that satisfies the algebraic conditions in \eqref{eq: lem:barrier-cert1}-\eqref{eq: lem:barrier-cert3}. 

Later, we develop a method for automatically deriving a finite collection of such algebraic conditions for the verification of temporal logic specifications which has been demonstrated to be an
appropriate specification formalism for reasoning about various kinds of systems \cite{Galton87TemporalLogics}.

\subsection{Linear Temporal Logic}
\label{ssec:ltl}
We employ linear temporal logic without the next operator 
($\LTLX$) to describe behaviors of continuous systems.%
\footnote{Similar to \cite{Kloetzer08}, our choice of $\LTLX$ over the widely used linear temporal logic that includes the next operator
is motivated by our definition of the satisfaction of a formula with discrete time semantics by a continuous trajectory.}
An $\LTLX$ formula is built up from a set of \emph{atomic propositions}
and two kinds of operators: logical connectives and temporal modal operators.
The logical connectives are those used in propositional logic: 
{\em negation\/} ($\neg$), {\em disjunction\/} (\hspace{-1mm}$\oor$\hspace{-1mm}), 
{\em conjunction\/} (\hspace{-1mm}$\aand$\hspace{-1mm}) and
{\em material implication\/} ($\Longrightarrow$).
The temporal modal operators include {\em always\/} ($\always$), 
{\em eventually\/} ($\eventually$) and {\em until\/} (\hspace{-1mm}$\until$\hspace{-1mm}).

\begin{defn}
An \emph{$\LTLX$ formula} over a set $\AP$ of atomic propositions is inductively defined as follows:
\begin{enumerate}[(1)]
\item $\true$ is an $\LTLX$ formula,
\item any atomic proposition $p \in \AP$ is an $\LTLX$ formula, and
\item given $\LTLX$ formulas $\varphi_1$ and $\varphi_2$, the formulas
  $\neg \varphi_1$, $\varphi_1 \oor \varphi_2$, and
  $\varphi_1 \until \varphi_2$ are also $\LTLX$ formulas.
\end{enumerate}
Additional operators can be derived from the logical connectives $\oor$ and $\neg$ and
the temporal modal operator \hspace{-1mm}$\until$\hspace{-1mm}.
For example, 
$\varphi_1 \aand \varphi_2 = \neg(\neg\varphi_1 \oor \neg \varphi_2)$, 
$\varphi_1 \Longrightarrow \varphi_2 = \neg \varphi_1 \oor \varphi_2$,
$\eventually\varphi = \true \until \varphi$ and
$\always\varphi = \neg\eventually\neg\varphi$.
\end{defn}

$\LTLX$ formulas are interpreted on infinite strings $\word = \alphabet_0 \alphabet_1 \alphabet_2 \ldots$ where $\alphabet_i \in 2^{\AP}$ for all $i \geq 0$.
Such infinite strings are referred to as {\em words\/}.
The satisfaction relation is denoted by $\models$, i.e.,
for a word $\word$ and an $\LTLX$ formula $\varphi$, we write $\word \models \varphi$ if and only if $\word$ satisfies $\varphi$
and write $\word \not\models \varphi$ otherwise.
The satisfaction relation is defined inductively as follows:
\begin{itemize}
\vspace{-1mm}
\item $\word \models \true$,
\vspace{-1mm}
\item for an atomic proposition $\ap \in \AP$, $\word \models \ap$ if and only if $\ap \in \alphabet_0$,
\vspace{-1mm}
\item $\word \models \neg \varphi$ if and only if $\word \not\models \varphi$,
\vspace{-1mm}
\item $\word \models \varphi_1 \aand \varphi_2$ if and only if $\word \models \varphi_1$ and $\word \models \varphi_2$, and
\vspace{-1mm}
\item $\word \models \varphi_1 \until \varphi_2$ if and only if there exists $j \geq 0$ such that 
$\alphabet_j \alphabet_{j+1} \ldots \models \varphi_2$ and for all $i$ such all $0 \leq i < j$, $\alphabet_i \alphabet_{i+1}\ldots \models \varphi_1$.
\end{itemize}

Given a proposition $\ap$, 
examples of widely used $\LTLX$ formulas include a safety formula of the form $\always \ap$ (read as ``always $\ap$'')
and a reachability formula of the form $\eventually \ap$ (read as ``eventually $\ap$'').
A word satisfies $\always \ap$ if 
$\ap$ remains invariantly true at all positions of the word
whereas it satisfies $\eventually \ap$ if
$\ap$ becomes true at least once in the word. 
By combining the temporal operators, we can express more complex properties.
For example $\always \eventually \ap$ states that $\ap$ holds infinitely often in the word.

Let $\varphi$ be an $\LTLX$ formula over $\AP$. The linear-time property induced by $\varphi$ is defined as
$Words(\varphi) = \{\word \in (2^{\AP})^\omega \ | \ \word \models \varphi\}$.

\subsection{Correctness of Dynamical Systems}
As described in Section \ref{ssec:ltl}, $\LTLX$ formulas are interpreted on infinite strings.
In this section, we show that the properties of trajectories of continuous systems can be characterized by such infinite strings,
allowing $\LTLX$ formulas to be interpreted over continuous trajectories.

The behavior of the system is formalized by a set $\AP$ of atomic propositions where
each atomic proposition $\ap \in \AP$ corresponds to a region of interest $\llbracket \ap \rrbracket \subseteq \X$.
Following \cite{Kloetzer08,LOTM13}, we define a trace of a trajectory to be the sequence of sets of propositions satisfied along the trajectory.
Specifically, for each $\alphabet \in 2^{\AP}$, we define
\begin{equation}
\llbracket \alphabet \rrbracket = 
\left\{ \begin{array}{ll}
\X \setminus \bigcup_{\ap \in \AP} \llbracket \ap \rrbracket &\hbox{if } a = \emptyset\\
\bigcap_{\ap \in \alphabet} \llbracket \ap \rrbracket \setminus \bigcup_{\ap \in \AP \setminus \alphabet} \llbracket \ap \rrbracket &\hbox{otherwise.}
\end{array} \right.
\label{eq:regions}
\end{equation}

According to Equation (\ref{eq:regions}), $\llbracket \emptyset \rrbracket$ is the subset of $\X$ that does not satisfy any atomic proposition in $\AP$ whereas for any $\alphabet \in 2^\AP$ such that $\alphabet \not= \emptyset$, $\llbracket \alphabet \rrbracket$ is the subset of $\X$ that satisfy all and only propositions in $\alphabet$.

\begin{defn}
\label{def:trace}
An infinite sequence $\trace{x} = \alphabet_0 \alphabet_1 \alphabet_2 \ldots$ where $\alphabet_i \in 2^{\AP}$ for all $i \in \naturals$
is a \emph{trace} of a trajectory $x: \reals_{\geq 0} \to \X$ of $\DS$ if
there exists an associated sequence $t_0 t_1 t_2 \ldots$ of time instances such that 
$t_0 = 0$, $t_k \to \infty$ as $k \to \infty$ and for each $i \in \naturals$, $t_i \in \reals_{\geq 0}$ satisfies the following conditions:
\begin{enumerate}[(1)]
\item $t_i < t_{i+1}$,
\item $x(t_i) \in \llbracket \alphabet_i \rrbracket$, and
\item \label{def:trace:tau}
if $\alphabet_i \not= \alphabet_{i+1}$, then for some $t_i' \in [t_i, t_{i+1}]$, 
$x(t) \in \llbracket \alphabet_i \rrbracket$ for all $t \in (t_i, t_i')$,
$x(t) \in \llbracket \alphabet_{i+1} \rrbracket$  for all $t \in (t_i', t_{i+1})$ and
either $x(t_i') \in \llbracket \alphabet_i \rrbracket$ or $x(t_i') \in \llbracket \alphabet_{i+1} \rrbracket$.
\end{enumerate}
\end{defn}



%

%
See Figure \ref{fig:trace} for a hypothetical example which explains the relation between a sample trajectory $x$ and its trace $\trace{x}$. 
In this case, we have $x(t_0), x(t_2), x(t_4) \in \X \setminus \bigcup_{\ap \in \AP} \llbracket \ap \rrbracket$,
$x(t_1) \in \llbracket p_A \rrbracket$, $x(t_3) \in \llbracket p_B \rrbracket$ and $x(t_5) \in \llbracket p_C \rrbracket$.
Hence, $\trace{x}$ is given by $\trace{x} = \emptyset \{p_A\} \emptyset \{p_B\} \emptyset \{p_C\} \ldots$.
Note that definition \ref{def:trace} is consistent with the definition of the word produced by a continuous trajectory in \cite{Kloetzer08,LOTM13} with slight differences. 
Specifically, the definition in \cite{Kloetzer08} has an additional requirement that if for any $i \in \naturals$, $\alphabet_i = \alphabet_{i+1}$, then $\llbracket \alphabet_i \rrbracket$ has to be a ``sink'' for the trajectory, i.e., $x(t) \in \llbracket \alphabet_i \rrbracket$ for all $t \geq t_i$.
Reference \cite{LOTM13} requires the time sequence $t_0 t_1 t_2 \ldots$ in Defition \ref{def:trace} to be exactly the instances where the sets of propositions satisfied by the trajectory changes, i.e., $t_i = \inf\{ t \ | \ t > t_{i-1}, x(t) \not\in \llbracket a_{k-1} \rrbracket\}$ for all $i > 0$.
We refer the reader to \cite{LOTM13} for the discussion on the existence of traces of realistic trajectories (i.e., those of \emph{finite variability}).

\begin{figure}
\centering
\includegraphics[width=0.25\textwidth]{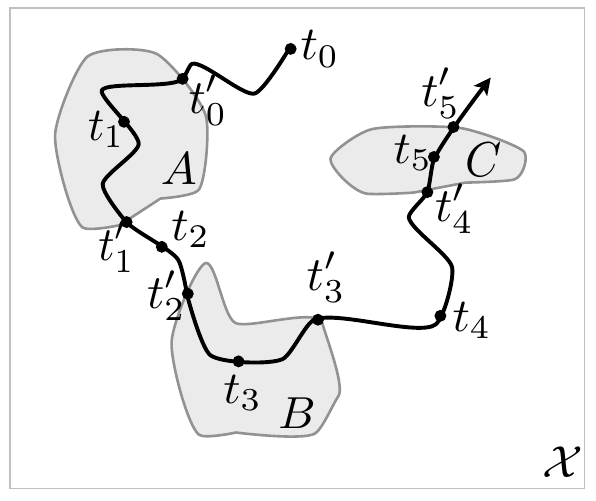}
\caption{{A hypothetical example which explains the relation between a sample trajectory and its trace.
As shown, $A, B, C \subset \X$.
Let $\AP = \{p_A, p_B, p_C\}$ where for each $S \in \{A, B, C\}$, $\llbracket p_S \rrbracket = S$.
The trajectory $x$ is represented by a solid curve starting at $t_0$.
A time sequence $t_0 t_1 t_2 \ldots$ associated with a trace of $x$ as well as
the intermediate time instances $t_0', t_1', t_2', \ldots$ satisfying condition \ref{def:trace:tau} of Definition \ref{def:trace} are as shown.
}}
\label{fig:trace}
\end{figure}

An important feature of a trace is that 
it captures the instances where the characteristics of the states along the trajectory (as defined by a combination of atomic propositions in $\AP$) change.
That is, a trace of $x$ characterizes the behavior of $x$ according to the sequence of sets of propositions satisfied, which correspond to regions visited,
along the trajectory.
%
Finally, define $Trace(\DS) = \{\trace{x} \in (2^{\AP})^\omega \ | \ $ there exists a trajectory $x$ of $\DS$ such that $\trace{x}$ is a trace of $x\}$ to be the set of traces of trajectories of $\DS$.

Next, we provide the definition of the satisfaction of an $\LTLX$ formula by $\DS$.

\begin{defn}
Given a trajectory $x$ of a dynamical system $\DS$ and an $\LTLX$ formula $\varphi$ over $\AP$,
we say that $x$ \emph{satisfies} $\varphi$ 
if for each infinite string $\trace{x} \in (2^{\AP})^\omega$ that is a trace of $x$, $\trace{x} \models \varphi$, i.e.,
the behavior of $x$ as captured by its trace is correct with respect to $\varphi$.
\end{defn}

\begin{defn}
A dynamical system $\DS$ \emph{satisfies} $\varphi$ if all trajectories of $\DS$ satisfy $\varphi$,
i.e., $Trace(\DS) \subseteq Words(\varphi)$.
\end{defn}

\subsection{Automata Representation of $\LTLX$ Formulas}
There is a tight relationship between $\LTLX$ 
and finite state automata that will be exploited in this paper.

\begin{defn}
\label{def:nba}
A \emph{non-deterministic Buchi automaton} (NBA) is a tuple
$\A = (Q, \Sigma, \delta, Q_0, F)$ where
\begin{itemize}
\item $Q$ is a finite set of states,
\item $\Sigma$ is a finite set, called an alphabet,
\item $\delta \subseteq Q \times \Sigma \times Q$ is a transition relation,
\item $Q_0 \subseteq Q$ is a set of initial states, and
\item $F \subseteq Q$ is a set of accepting (or final) states.
\end{itemize}
We use the relation notation, $q \stackrel{\alphabet}{\longrightarrow} q'$,
to denote $(q,\alphabet,q') \in \delta$.
\end{defn}

Consider an NBA $\A = (Q, \Sigma, \delta, Q_0, F)$.
%
%
Let $\pi$ be a sequence of states of $\A$, i.e.,
$\pi = q_0q_1 \ldots q_m$ for some $m \in \naturals,$ if it is finite,
and $\pi = q_0q_1 \ldots$ where $q_i \in Q$ for all $i$, if it is infinite.
We say that $\pi$ is a \emph{run fragment} of $\A$ if, for each $i$,
there exists $\alphabet_i \in \Sigma$ such that $q_i \stackrel{\alphabet_i}{\longrightarrow} q_{i+1}$.
Hence, {a finite run fragment $\pi = q_0q_1 \ldots q_m$ of $\A$ generates a set 
$\ST(\pi) = \{\alphabet_0 \alphabet_1 \ldots \alphabet_{m-1} \in \Sigma^* \ | \ q_i \stackrel{\alphabet_i}{\longrightarrow} q_{i+1} \hbox{ for all } i \in \{0, \ldots, m-1\}\}$
of finite strings and an infinite run fragment $\pi = q_0q_1 \ldots$ generates a set 
$\ST(\pi) = \{\alphabet_0 \alphabet_1 \ldots \in \Sigma^\omega \ | \ q_i \stackrel{\alphabet_i}{\longrightarrow} q_{i+1} \hbox{ for all } i\}$} of infinite strings.
A \emph{run} of $\A$ is an infinite run fragment $\pi = q_0q_1 \ldots$
such that $q_0 \in Q_0$.
Given an infinite string $\word = \alphabet_0\alphabet_1 \ldots \in \Sigma^\omega$,
a \emph{run for $\word$} in $\A$ is an infinite sequence of states
$\pi = q_0q_1\ldots$ such that $q_0 \in Q_0$ and 
$q_i \stackrel{\alphabet_i}{\longrightarrow} q_{i+1}$ for all $i \geq 0$, i.e.,
$\word \in \ST(\pi)$.
A run is \emph{accepting} if there exist infinitely many $j \geq 0$ such that $q_j \in F$.
A string $\word \in \Sigma^\omega$ is \emph{accepted} by $\A$ if there is an accepting run for $\word$ in $\A$.
The \emph{language} accepted by $\A$, denoted by $\lang(\A)$, is the set of all accepted strings of $\A$.
%

It can be shown that for any $\LTLX$ formula $\varphi$ over $\AP$, 
there exists an NBA $\A_\varphi$ with alphabet $\Sigma = 2^{\AP}$ that accepts all words and only those words over $\AP$ that satisfy $\varphi$,
i.e., $\lang(\A_\varphi) = Words(\varphi) = \{\word \in (2^{\AP})^\omega \ | \ \word \models \varphi\}$ 
\cite{Baier:PMC2008,he08-ltltobuchi,gastin01-fastltl}.
Such $\A_\varphi$ can be automatically constructed using existing tools, such as
LTL2BA \cite{ltl2ba}, SPIN \cite{spin} and LBT \cite{lbt},
with the worst-case complexity that is exponential in the length of $\varphi$.


\section{Problem Formulation}
\label{sec:prob}
Consider a dynamical system $\DS$ of the form (\ref{eq:sys}) 
and a set $\AP = \{p_0, p_1, \ldots, p_N\}$ of atomic propositions.
For each atomic proposition $p_i$, we let $\X_i = \llbracket p_i \rrbracket \subseteq \X$ denote the set of states that satisfy $p_i$. \\

\noindent {\bf Problem statement:} Given a specification $\varphi$ expressed as an $\LTLX$ formula over $\AP$, determine if $\DS$ satisfies $\varphi$.

\begin{figure}[t]
\centering
\includegraphics[width=0.35\textwidth]{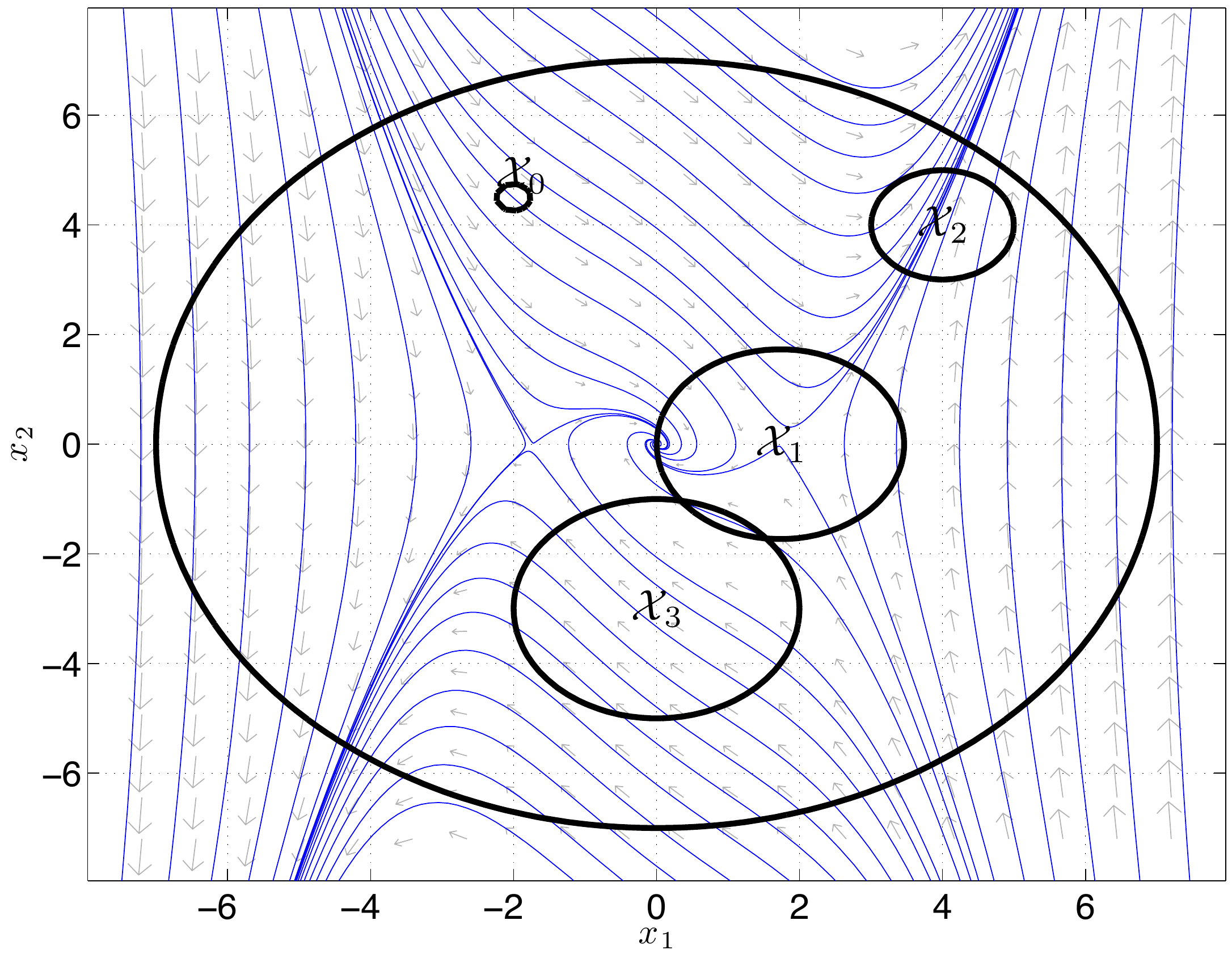}
\caption{Phase portrait of the dynamical system in (\ref{ex1:dyn}), some representative trajectories (blue curves), and the sets $\X$, $\X_0, \ldots, \X_3$ defined in (\ref{ex1:sets}). Thick (black) curves are the boundaries of $\X$, $\X_0, \ldots, \X_3$ with the biggest circle being
the boundary of $\X$.}
\label{f:ex1-phase-portrait}
\end{figure}

\begin{exmp}
\label{ex:running-ex}
We use a simple problem to demonstrate the main ideas throughout the paper. Consider a two-dimensional system (which also appears in \cite{Khalil:Nonlinear96,prajna05thesis}) governed by
\begin{equation}
\label{ex1:dyn}
\begin{array}{rcl}
  \dot{x}_1(t) &=& x_2(t)\\
  \dot{x}_2(t) &=& -x_1(t) + \frac{1}{3}x_1(t)^3 - x_2(t),
\end{array}
\end{equation}
over the domain $\X = \{(x_1, x_2) \ | \ x_1^2 + x_2^2 \leq 49\}$ and let the regions of interest be given as
\begin{equation}
\label{ex1:sets}
\begin{array}{rcl}
  \X_0 &=& \left\{(x_1, x_2) \ | (x_1 + 2)^2 + (x_2 - 4.5)^2 \leq 0.0625\right\},\\
  \X_1 &=& \left\{(x_1, x_2) \ | (x_1 - \sqrt{3})^2 + x_2^2 \leq 3\right\},\\
  \X_2 &=& \left\{(x_1, x_2) \ | (x_1 - 4)^2 + (x_2 - 4)^2 \leq 1\right\}, and\\
  \X_3 &=& \left\{(x_1, x_2) \ | x_1^2 + (x_2 + 3)^2 \leq 4\right\}.
\end{array}
\end{equation}

The phase portrait of (\ref{ex1:dyn}) and the sets $\X$, $\X_0, \ldots, \X_3$ are shown in Figure~\ref{f:ex1-phase-portrait}.
In this case, $\AP = \{p_0, p_1, \ldots, p_3\}$, where for each $i \in \{0, \ldots, 3\}$, 
$\llbracket p_i \rrbracket = \X_i$.

We want to ensure that any trajectory of (\ref{ex1:dyn}) satisfies the following conditions.
\begin{itemize}
\vspace{-1mm}
\item Once it reaches $\X_2$, it cannot reach $\X_3$ forever.
\vspace{-1mm}
\item If it starts in $\X_0$, then it has to reach $\X_1$ before it reaches $\X_2$.
\end{itemize}

\vspace{-1mm}
The property described above can be expressed as the $\LTLX$ formula
\begin{equation}
\label{ex1:phi}
  \varphi = \always(p_2 \implies \always \neg p_3) \aand \big (p_0 \implies 
  (\eventually p_2 \implies (\neg p_2 \until p_1)) \big).
\end{equation}


\end{exmp}

\section{Automata-Based Verification}
\label{sec:approach}

Our approach to solve the $\LTLX$ verification of dynamical systems
defined in Section \ref{sec:prob}  relies on constructing a set 
$\Omega \subseteq (2^{\AP})^*$ of finite strings {such that
for any word $\word \in (2^{\AP})^\omega$, if $\word \not\models \varphi$,
then there exists a substring $\omega \in \Omega$ of $\word$.}
Hence, to provide a proof of correctness of $\DS$ with respect to $\varphi$,
we ``invalidate'' each $\omega \in \Omega$ by showing that
$\omega$ cannot be a substring of any word in $Trace(\DS)$.

To compute the set $\Omega$, we first generate an NBA $\A_{\neg\varphi} = (Q, 2^{\AP}, \delta, Q_0, F)$ 
that accepts all words and only those words over $\AP$ that satisfy $\neg\varphi$.
It is well known from automata theory and model checking \cite{Baier:PMC2008} that $Trace(\DS) \not\subseteq Words(\varphi)$
if and only if there exists a word in $Trace(\DS)$ that is accepted by $\A_{\neg\varphi}$.
Furthermore, there exists a word $\word \in (2^{\AP})^\omega$ that is accepted by $\A_{\neg\varphi}$ if and only if
there exists a run of $\A_{\neg\varphi}$ of the form 
$q^p_0 q^p_1 \ldots q^p_{m_p}(q^c_0 q^c_1 \ldots q^c_{m_c})^\omega$
where $m_p, m_c \in \naturals$ and $q^c_0 \in F$.
%
%

Let $\R^{fin}$ be the set of finite run fragments of $\A_{\neg\varphi}$.
In addition, for each $q, q' \in Q$, let $\R(q, q') \subseteq \R^{fin}$ be the set of
finite run fragments of $\A_{\neg\varphi}$ that starts in $q$ and ends in $q'$.
%
Consider the set $\R^{acc} \subseteq \R^{fin}$ defined by
$\R^{acc} = \{\pi^p \pi^c \ | \ \pi^p \in \R(q_0, q), \pi^c \in \R(q', q), q_0 \in Q_0, q \in F, 
q \stackrel{\alphabet}{\longrightarrow} q' \hbox{ for some } \alphabet \in 2^{\AP}\}$.
Note that any run fragment in $\R^{acc}$ consists of two parts, $\pi^p$ and $\pi^c$,
where $\pi^p$ corresponds to a finite run fragment from an initial state to an accepting state $q$ of $\A_{\neg\varphi}$
and $q \pi^c$ corresponds to a finite run fragment from and to $q$, i.e., an accepting cycle starting with $q$.
Finally, define $\Omega$ as the set of all finite strings generated by run fragments in $\R^{acc}$, i.e.,
$\Omega = \bigcup_{\pi \in \R^{acc}} \ST(\pi)$.

\begin{exmp}
\label{ex:substring}
Figure~\ref{fig:ex-aut} shows an NBA $\A_{\neg\varphi}$ that accepts all and only words that satisfy 
$\neg\varphi$ where $\varphi$ is defined in (\ref{ex1:phi}). 
Note that the transitions are simplified and only valid transitions, i.e., transitions $(q, \alphabet, q')$ such that
$\llbracket \alphabet \rrbracket \not= \emptyset$ are shown.
%
From Figure \ref{fig:ex-aut}, we get that $Q_0 = \{q_0\}$ and $F = \{q_4\}$.
Hence, the set of run fragments from initial states to accepting states of $\A_{\neg\varphi}$ is given by
$\R(q_0, q_4) = \{q_0 q_1^+ q_4^+, q_0 q_2^+ q_3^+ q_4^+,$ $q_0 q_3^+ q_4^+\}$
and the set of accepting cycles of $\A_{\neg\varphi}$ is given by $\R(q_4, q_4) = \{q_4^+\}$.
By appending run fragments in $\R(q_4, q_4)$ to those in $\R(q_0, q_4)$,
we obtain $\R^{acc} = \{q_0 q_1^+ q_4 q_4^+, q_0 q_2^+ q_3^+ q_4 q_4^+, q_0 q_3^+ q_4 q_4^+\}$.
$\Omega$ is then defined as the union of the following sets of finite strings:
\begin{itemize}
\vspace{-1mm}
\item $\{\alphabet_{0,1}\alphabet_{1,1}^1 \ldots \alphabet_{1,1}^k\alphabet_{1,4}\alphabet_{4,4}^1 \ldots \alphabet_{4,4}^l \ | \  k \geq 0, l > 0,
p_0 \in \alphabet_{0,1}, p_2 \in \alphabet_{1,4}, p_1 \not\in \alphabet_{1,1}^j \hbox{ for all } j \in \{1, \ldots k\} \}$, 
which is generated by $q_0 q_1^+ q_4 q_4^+$,
\vspace{-1mm}
\item $\{\alphabet_{0,2}\alphabet_{2,2}^1 \ldots \alphabet_{2,2}^{k_1} \alphabet_{2,3}\alphabet_{3,3}^1 \ldots \alphabet_{3,3}^{k_2}\alphabet_{3,4}\alphabet_{4,4}^1 \ldots \alphabet_{4,4}^l \ | \  
k_1, k_2 \geq 0, l > 0,
p_2 \in \alphabet_{2,3}, p_3 \in \alphabet_{3,4}\}$, 
which is generated by $q_0 q_2^+ q_3^+ q_4 q_4^+$, and
\vspace{-1mm}
\item $\{\alphabet_{0,3}\alphabet_{3,3}^1\ldots\alphabet_{3,3}^k\alphabet_{3,4}\alphabet_{4,4}^1\ldots\alphabet_{4,4}^l \ | \ k \geq 0, l > 0,
p_2 \in \alphabet_{0,3}, p_3 \in \alphabet_{3,4}\}$,
which is generated by $q_0 q_3^+ q_4 q_4^+$.
\end{itemize}
\end{exmp}

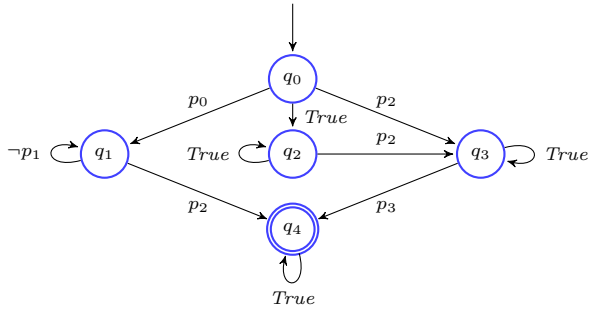
\begin{figure}
\centering
\begin{tikzpicture}[->,>=stealth',shorten >=1pt,auto,node distance=2cm, bend angle=15, font=\scriptsize]
  \tikzstyle{every state}=[circle,thick,draw=blue!75,minimum size=6mm]

   \node [state] (q0) at (0,0) {$q_0$};
   \node [state] (q1) at (-2.5,-1){$q_1$};
   \node [state] (q2) at (0,-1) {$q_2$};
   \node [state] (q3) at (2.5,-1) {$q_3$};
   \node [state, double] (q4) at (0,-2){$q_4$};
   
   \draw [ shorten >= 1pt, -> ] (0,1) to (q0);
   
   \path (q0) edge [] node [above] {$p_0$} (q1)
                     edge [] node [] {$\true$} (q2)
                     edge [] node [above] {$p_2$} (q3)
             (q1) edge [loop left] node [] {$\neg p_1$} (q1)
                     edge [] node [below] {$p_2$} (q4)
             (q2) edge [loop left] node [] {$\true$} (q2)
                     edge node [above] {$p_2$} (q3)
             (q3) edge [loop right] node [] {$\true$} (q3)
                     edge [] node [below] {$p_3$} (q4)
             (q4) edge [loop below] node [] {$\true$} (q4);
\end{tikzpicture}
\caption{NBA $\A_{\neg\varphi}$ that accepts all and only words that satisfy 
$\neg\varphi$ where $\varphi$ is defined in (\ref{ex1:phi}).
Note that the transitions are simplified and only valid transitions, i.e., transitions $(q, \alphabet, q')$ such that
$\llbracket \alphabet \rrbracket \not= \emptyset$ are shown.
For example, the transition $(q_0, p_0 \aand \neg p_1, q_1)$ is labeled with $p_0$ because $\X_0 \cap (\X \setminus \X_1) = \X_0$.
An arrow without a source points to an initial state. An accepting state is drawn with a double circle.}
\label{fig:ex-aut}
\end{figure}

\vspace{-1mm}
\begin{lem}
\label{lem:omega}
For any infinite string $\word \in (2^{\AP})^\omega$, if $\word \not\models \varphi$,
then there exists a substring $\omega \in \Omega$ of $\word $.
\end{lem}
\begin{proof}
Consider an infinite string $\word = \alphabet_0 \alphabet_1 \ldots \in (2^{\AP})^\omega$ such that $\word \not\models \varphi$.
From automata theory \cite{Baier:PMC2008}, 
$\word \in \lang(\A_{\neg\varphi})$; hence, there exists an accepting run $\pi = q_0 q_1 \ldots$ for $\word$ in $\A_{\neg\varphi}$.
Since $\pi$ is an accepting run, by definition, there exists $q \in F$ such that $q_i = q$ for infinitely many $i$.
Let $j \geq 0$ and $k > j$ be indices such that $q_j = q_k = q$ and
consider $\omega = \alphabet_0 \alphabet_1 \ldots \alphabet_{k-1}$.
Clearly, $\omega$ is a substring of $\word$.
Furthermore, $q_0 q_1 \ldots q_j \in \R(q_0, q)$ and $q_{j+1} q_{j+2} \ldots q_k \in \R(q', q)$ where
$q \stackrel{\alphabet_j}{\longrightarrow} q'$.
Thus, it is clear from the definition of $\R^{acc}$ that $\pi' = q_0q_1 \ldots q_k \in \R^{acc}$.
Since $\omega \in \ST(\pi')$, we can conclude that $\omega \in \Omega$.
%
\end{proof}

\begin{exmp}
{
Consider an infinite string $\word = \alphabet_0 \alphabet_1 \ldots$ such that $p_2 \in \alphabet_i$
for some $i \in \naturals$ and $p_3 \in \alphabet_j$ for some $j > i$.
It is obvious that $\word \not\models \always(p_2 \implies \always \neg p_3)$; hence,
$\word \not\models \varphi$ where $\varphi$ is defined in (\ref{ex1:phi}).
Based on Lemma \ref{lem:omega}, there must exist a substring $\omega \in \Omega$ of $\word$.
Consider a substring 
$\omega = \alphabet_{0,3}\alphabet_{3,3}^1\ldots\alphabet_{3,3}^{j-i-1}\alphabet_{3,4} \alphabet_{4,4}$ of $\word$
where $\alphabet_{0,3} = \alphabet_i$, $\alphabet_{3,3}^1 = \alphabet_{i+1}, \ldots, \alphabet_{3,3}^{j-i-1} = \alphabet_{j-1}$, 
$\alphabet_{3,4} = \alphabet_j$ and $\alphabet_{4,4} = \alphabet_{j+1}$.
It is easy to check that
$\omega \in \{\alphabet_{0,3}\alphabet_{3,3}^1\ldots\alphabet_{3,3}^k\alphabet_{3,4}\alphabet_{4,4}^1\ldots\alphabet_{4,4}^l \ | \ k \geq 0, l > 0,
p_2 \in \alphabet_{0,3}, p_3 \in \alphabet_{3,4}\}$;
hence from Example \ref{ex:substring}, $\omega \in \Omega$.
}
\end{exmp}

\begin{lem}
\label{lem:substring-omega}
Suppose for each $\omega \in \Omega$,
there exists a substring $\omega'$ of $\omega$ such that 
$\omega'$ cannot be a substring of any word in $Trace(\DS)$.
Then, $\DS$ satisfies $\varphi$.
\end{lem}
\begin{proof}
Assume, in order to establish a contradiction, that $\DS$ does not satisfy $\varphi$.
Then, there exists a trajectory $x$ of $\DS$ and its trace $\trace{x}$ such that
$\trace{x} \not\models \varphi$.
From Lemma \ref{lem:omega}, there exists a substring $\omega \in \Omega$ of $\trace{x}$.
However, since $\omega \in \Omega$, there exists a substring $\omega'$ of $\omega $
that is not a substring of $\trace{x}$.
Hence, $\omega$ cannot be a substring of $\trace{x}$, leading to a contradiction.
\end{proof}

Based on Lemma \ref{lem:substring-omega},
we can verify that $\DS$ satisfies $\varphi$ by checking that for each
$\omega \in \Omega$, there exists a substring of $\omega$
that cannot be a substring of any word in $Trace(\DS)$.
However, since $\R^{acc}$ is, in general, not finite, $\Omega$ is also, in general, not finite (as illustrated in Example \ref{ex:substring}).
As a result, invalidating all $\omega \in \Omega$ may not be straightforward.
In the next section, we propose a finite collection $\Pi_1, \Pi_2, \ldots, \Pi_M$ of representative sets of finite run fragments
with the property that for each $\omega \in \Omega$, there exists some $i \in \{1, \ldots, M\}$, such that each
$\pi \in \Pi_i$ can be used to ``derive'' a substring of $\omega$ that is in a certain form.
(We will make it clear later how such a substring can be derived.)
Hence, invalidating all strings derived from some $\pi \in \Pi_i$ for each $i \in \{1, \ldots, M\}$ provides a certificate of system correctness with respect to $\varphi$.
Then, in Section \ref{sec:bc}, we show that due to their particular form, the strings derived from any $\pi \in \Pi_i$, $i \in \{1, \ldots, M\}$
are amenable to verification based on the idea of barrier certificates and to algorithmic solutions, for the cases where the vector field in \eqref{eq:sys}
and the sets $\X, \X_0, \ldots, \X_N$ can be described by 
polynomial or rational functions, through sum-of-squares relaxations for polynomial optimization.

To recap, based on the definition of a trace, the behavior of $\DS$ is formalized by the sequences of subsets of $\X$ visited along its trajectories.
These subsets of $\X$ are constructed from a collection of sets $\X_1, \ldots, \X_N$; hence,
each of them captures certain characteristics of $\DS$ over $\X$ as described by a boolean combination of atomic propositions in $\AP$.
The language $\lang(\A_{\neg\varphi})$ accepted by $\A_{\neg\varphi}$ essentially describes the sequences of subsets of $\X$ that violate $\varphi$.
Hence, to prove that $\DS$ satisfies $\varphi$, 
we show that for each of its trajectories and for each sequence in $\lang(\A_{\neg\varphi})$,
there exists a portion of the sequence that the trajectory cannot follow.


\section{Representative Sets of Run Fragments}
\label{sec:substrings}

Let $\G = (V^\G, E^\G)$ denote the underlying directed graph of $\A_{\neg\varphi}$, i.e.,
$V^\G = Q$ and $E^\G \subseteq V^\G \times V^\G$ such that $(q, q') \in E^\G$
if and only if there exists $\alphabet \in 2^{\AP}$ such that $q \stackrel{\alphabet}{\longrightarrow} q'$.
A path in $\G$ is a finite or infinite sequence $\pi$ of states such that for any two consecutive states $q, q'$
in $\pi$, $(q, q') \in E^\G$.
From the construction of $\G$, it is obvious that $\pi$ is a path in $\G$ if and only if it is a run fragment of $\A_{\neg\varphi}$.
Given a finite path $\pi = q_0 q_1 \ldots q_m$ or an infinite path $\pi = q_0 q_1 \ldots$,
a subpath of $\pi$ is any finite path of the form $q_i q_{i+1} \ldots q_{i+k}$
where $i,k \geq 0$ and $i+k \leq m$ if $\pi$ is finite.

A variant of depth-first search {\cite{Russell:AI1995}} provided in Algorithm~\ref{alg:dfs}
can be used to find the set of all the paths from a state $q$ to a state $q'$ with no repeated edges and no consecutive repetitions of states in $\G$,
including the case where $q = q'$.
Since $E^\G$ is finite, the set of all the paths from $q$ to $q'$ with no repeated edges and no consecutive repetitions of states is finite for any $q, q' \in Q$
(unlike the set of all the paths from $q$ to $q'$ which may not be finite as these paths may contain cycles that can be repeated arbitrary times).
As will be discussed later, such a set of paths with no repeated edges and no consecutive repetitions of states
can be used to form a finite set $\SP$ of subpaths from $q$ to $q'$,
each of which can be ``extended'' to a subpath of any path from $q$ to $q'$.
Proposition \ref{lem:conds}, presented later, provides an exact definition of ``extending'' a path.

\begin{algorithm}[t]
\label{alg:dfs}

\begin{algorithmic}[1]

\State $\P^\G_{q,q'} \gets \emptyset$
\State $\mathit{toVisit} \gets \{q\}$
\State $\mathit{paths} \gets \{q\}$

\If {$q' = q$}
\State Append $q$ to $\P^\G_{q,q'}$
\EndIf

\While{$\mathit{toVisit} \not= \emptyset$}
\State Remove the last element of $\mathit{toVisit}$ and assign it to $v$
\State Remove the last sequence in $\mathit{paths}$ and assign it to $\mathit{path2v}$
	
\ForAll{$\textit{nb} \not= v$ such that $(v, \textit{nb}) \in E^\G$}
\If{$\textit{nb} = q'$}
\State Append the sequence obtained by concatenating $\mathit{path2v}$ and $nb$ to $\P^\G_{q,q'}$
\ElsIf{$v$ is not followed by $\textit{nb}$ in $\textit{path2v}$}
\State Append $\textit{nb}$ to $\textit{toVisit}$
\State Append the sequence obtained by concatenating $\textit{path2v}$ and $nb$ to $\textit{paths}$;
\EndIf
\EndFor
\EndWhile

\State \textbf{return} $\P^\G_{q,q'}$

\end{algorithmic}
\caption{\texttt{DFS}$(\G, q, q')$ }
\end{algorithm}

Given $q, q' \in Q$, let $\P(q,q')$ be the set of paths from a state $q$ to a state $q'$ with no repeated edges and no consecutive repetitions of states in $\G$.
In addition, for each $q \in F$, 
let $\P^{path}(q) = \{\pi \in \P(q_0, q) \ | \ q_0 \in Q_0\}$ be the set of paths from an initial state of $\A_{\neg\varphi}$ to $q$
with no repeated edges and no consecutive repetitions of states and
let $\P^{cyc}(q) = \{\pi \in \P(q,q) \ | \  {\P^{path}(q)} \not= \emptyset \hbox{ and if } \pi = q \hbox{, then } (q,q) \in E^\G\}$ 
be the set of reachable cycles that start from $q$ and have no repeated edges or consecutive repetitions of states.
From the definition of $\P(\cdot, \cdot)$ and $\R(\cdot, \cdot)$, it is obvious that for each $q \in F$,
$\P^{cyc}(q)$ and $\P^{path}(q)$ are finite, 
$\P^{cyc}(q) \subseteq \R(q, q)$ and $\P^{path}(q) \subseteq \bigcup_{q_0 \in Q_0} \R(q_0, q)$.
In this section, we show that a collection $\Pi_1, \Pi_2, \ldots, \Pi_M$ of representative sets of finite run fragments 
as described at the end of Section \ref{sec:approach} can be constructed from $\P^{cyc}(q)$ and $\P^{path}(q)$ for each $q \in F$.

{For a finite path $\pi$ in $\G$, we define $\PF^3(\pi)$ as the set of all subpaths of $\pi$ with length 3, i.e., 
$\PF^3(q_0 q_1 \ldots q_m) = \{q_i q_{i+1} q_{i+2} \ | \ 0 \leq i \leq m-2\}$.} 
Note that for a path $\pi$ with length less than 3, $\PF^3(\pi) = \emptyset$.

\begin{exmp}
\label{ex:paths}
Let $\A_{\neg\varphi}$ be the NBA shown in Figure \ref{fig:ex-aut}.
Then, $F = \{q_4\}$.
Applying Algorithm \ref{alg:dfs}, we get
\begin{eqnarray*}
\P^{cyc}(q_4) &=& \{q_4\},\\
\P^{path}(q_4) &=& \{q_0q_1q_4, q_0q_2q_3q_4, q_0q_3q_4\},\\
\PF^3(q_4) &=& \emptyset,\\
\PF^3(q_0q_1q_4) &=& \{q_0q_1q_4\},\\
\PF^3(q_0q_2q_3q_4) &=& \{q_0q_2q_3,q_2q_3q_4\},\\
\PF^3(q_0q_3q_4) &=& \{q_0q_3q_4\}.
\end{eqnarray*}
\end{exmp}


Note that any $\omega \in \Omega$ can be written as $\omega = \omega^p \omega^c$ where
$\omega^p$ and $\omega^c$ are generated from $\pi^p$ and $q \pi^c$, respectively, for some $\pi^p \pi^c \in \R^{acc}$ 
where $\pi^p$ corresponds to a finite run fragment from an initial state to an accepting state $q$ of $\A_{\neg\varphi}$
and $q \pi^c$ corresponds to an accepting cycle of $\A_{\neg\varphi}$. 
Hence, to invalidate $\omega$, we can invalidate either $\omega^p$ or $\omega^c$. 
As will be shown in Proposition \ref{lem:conds}, for any path $\pi$ from $q$ to $q'$,
there exists $\pi' \in \P(q, q')$ such that $\SP = \PF^3(\pi')$ is a finite set of paths, each of which can be extended to a subpath of $\pi$.
Hence, a way to invalidate $\omega^c$ is to show that for each $p \in \P^{cyc}(q)$, 
there exists $\tilde{\pi} \in \PF^3(p)$ such that 
all finite strings generated by each extension of $\tilde{\pi}$ cannot be a substring of any word of $\DS$.
Similarly, a way to invalidate $\omega^p$ is to show that for each $p \in \P^{path}(q)$,
there exists $\tilde{\pi} \in \PF^3(p)$ such that 
all finite strings generated by each extension of $\tilde{\pi}$ cannot be a substring of any word of $\DS$.

\begin{prop}
\label{lem:conds}
Suppose for each $q \in F$, either of the following conditions (1) and (2) holds:
\begin{enumerate}[(1)]
\item For each $p \in \P^{cyc}(q)$, there exists $\pi = q_0 q_1 q_2 \in \PF^3(p)$ such that 
\begin{enumerate}[(a)]
\item all finite strings $\alphabet_0\alphabet_1 \in \ST(\pi)$ cannot be a substring of any word in $Trace(\DS)$, and
\item if $(q_1, q_1) \in E^\G$, then all finite strings $\alphabet_0 \tilde{\alphabet}_0 \ldots \tilde{\alphabet}_k \alphabet_1 \in \ST(q_0q_1q_1^+q_2)$, $k \in \naturals$
cannot be a substring of any word in $Trace(\DS)$.
\end{enumerate}
\item For each $p \in \P^{path}(q)$, there exists $\pi = q_0 q_1 q_2 \in \PF^3(p)$ such that
\begin{enumerate}[(a)]
\item all finite strings $\alphabet_0\alphabet_1 \in \ST(\pi)$ cannot be a substring of any word in $Trace(\DS)$, and
\item if $(q_1, q_1) \in E^\G$, then all finite strings $\alphabet_0\tilde{\alphabet}_0 \ldots \tilde{\alphabet}_k\alphabet_1 \in \ST(q_0q_1q_1^+q_2)$, $k \in \naturals$
cannot be a substring of any word in $Trace(\DS)$.
\end{enumerate}
\end{enumerate}
Then, $\DS$ satisfies $\varphi$.
\end{prop}
\begin{proof}
Consider an arbitrary finite string $\omega \in \Omega$.
From the definition of $\Omega$, there exist an accepting state $q \in F$ and a finite run fragment of the form
$q^p_0 q^p_1 \ldots q^p_{m_p} q q^c_0 q^c_1 \ldots q^c_{m_c} q$
where $m_p, m_c \in \naturals$ and $q^p_0 \in Q_0$ from which $\omega$ is generated.
Let $\pi^p = q^p_0 q^p_1 \ldots q^p_{m_p} q$ and
$\pi^c = q q^c_0 q^c_1 \ldots q^c_{m_c} q$.
In addition, let $\omega^p$ and $\omega^c$ be the substrings of $\omega$ that are generated
from $\pi^p$ and $\pi^c$, respectively.
Note that both $\pi^p$ and $\pi^c$ correspond to paths in $\G$.
To prove that satisfying either condition (1) or (2) ensures the correctness of $\DS$ with respect to $\varphi$, we show that both of the following conditions hold.
\begin{enumerate}[(i)]
\item 
\label{item:pf:lem:conds1}
There exists $p \in \P^{cyc}(q)$ such that for each $\pi = q_0q_1q_2 \in \PF^3(p)$, 
if $(q_1, q_1) \not\in E^\G$, 
then $\pi^c$ contains $\pi$; 
otherwise $\pi^c$ contains some run fragment of the form $q_0q_1^+q_2$.
\item 
\label{item:pf:lem:conds2}
There exists $p \in \P^{path}(q)$ such that for each $\pi = q_0q_1q_2 \in \PF^3(p)$,
if $(q_1, q_1) \not\in E^\G$, 
then $\pi^p$ contains $\pi$; 
otherwise $\pi^p$ contains some run fragment of the form $q_0q_1^+q_2$.
\end{enumerate}
Thus, satisfying condition (1) ensures that there exists a substring $\omega^{c'}$
of $\omega^c$ such that $\omega^{c'}$ cannot be a substring of any word in $Trace(\DS)$.
Since $\omega^c$ is a substring of $\omega$, $\omega^{c'}$ is also a substring of $\omega$.
We can then conclude from Lemma \ref{lem:substring-omega} that $\DS$ satisfies $\varphi$.
Similarly, satisfying condition (2) ensures that there exists a substring $\omega^{p'}$
of $\omega^p$, which is also a substring of $\omega$, that cannot be a substring of any word in $Trace(\DS)$.
Lemma \ref{lem:substring-omega} can then be applied to conclude that $\DS$ satisfies $\varphi$.

First, consider condition (\ref{item:pf:lem:conds1}) and 
the case where $\pi^p$ does not contain any repeated edges or consecutive repetitions of states in $\G$.
In this case, it directly follows from the definition of $\P^{path}$ that
$\pi^p \in \P^{path}(q)$;
hence, condition (\ref{item:pf:lem:conds1}) is trivially satisfied.
Next, consider the case where $\pi^p$ contains a repeated edge, i.e., there exist $\tilde{q}_1, \tilde{q}_2 \in Q$ such that
$\tilde{q}_1$ is followed by $\tilde{q}_2$ more than once in $\pi^p$.
Then, $\pi^p$ must contain a subsequence of the form $\tilde{q}_1 \tilde{q}_2 \ldots \tilde{q}_1 \tilde{q}_2$.
Let $\pi^{p'}$ 
be a run fragment that is obtained from $\pi^p$ by replacing this subsequence with $\tilde{q}_1 \tilde{q}_2$; 
thus, removing a repeated edge $(\tilde{q}_1, \tilde{q}_2)$ from $\pi^p$.
It can be checked that for any $\pi = q_0q_1q_2 \in \PF^3(\pi^{p'})$,
if $(q_1, q_1) \not\in E^\G$, then $\pi \in \PF^3(\pi^p)$; otherwise,
$\pi^p$ contains a subsequence of the form $q_0q_1^+q_2$.
For the case where $\pi^p$ contains a consecutive repetition of some state $\tilde{q} \in Q$,
i.e., $\pi^p = q^p_0 q^p_1 \ldots \tilde{q}\tilde{q} \ldots \tilde{q} \ldots q^p_{m_p} q$,
we construct $\pi^{p''} = q^p_0 q^p_1 \ldots \tilde{q} \ldots q^p_{m_p} q$ by removing such a consecutive repetition of $\tilde{q}$.
It can be easily checked that for any $\pi = q_0q_1q_2 \in \PF^3(\pi^{p''})$,
if $(q_1, q_1) \not\in E^\G$, then $\pi \in \PF^3(\pi^p)$; otherwise,
$\pi^p$ contains a subsequence of the form $q_0q_1^+q_2$.
We apply this process of removing repeated edges and consecutive repetitions of states in $\pi^p$
until we obtain a run fragment $\tilde{\pi}^{p}$ that does not contain any repeated edges or consecutive repetitions of states.
Then, $\tilde{\pi}^{p} \in \P^{path}(q)$ and for any $\pi = q_0q_1q_2 \in \PF^3(\tilde{\pi}^{p})$,
if $(q_1, q_1) \not\in E^\G$, then $\pi \in \PF^3(\pi^p)$; otherwise,
$\pi^p$ contains a subsequence a of the form $q_0q_1^+q_2$.
Condition (\ref{item:pf:lem:conds2}) can be treated in a similar way.
\end{proof}

To sum, Proposition \ref{lem:conds} provides a sufficient (but not necessary) condition for verifying
that no word in $Trace(\DS)$ is accepted by $\A_{\neg\varphi}$.
Based on Proposition \ref{lem:conds}, we construct sets
$\PF_1^{cyc,q}, \PF_2^{cyc,q}, \ldots, \PF_{M_c}^{cyc,q}$ and
$\PF_1^{path,q}, \PF_2^{path,q}, \ldots,$ $\PF_{M_p}^{path,q}$
for each $q \in F$ where
$M_c$ is the cardinality of $\P^{cyc}(q)$,
$M_p$ is the cardinality of $\P^{path}(q)$,
for each $i \in \{1, \ldots, M_c\}$,
$\PF_i^{cyc,q} = \PF^3(p)$, $p$ is the $i$th path in $\P^{cyc}(q)$ and
for each $i \in \{1, \ldots, M_p\}$,
$\PF_i^{path,q} = \PF^3(p)$, $p$ is the $i$th path in $\P^{path}(q)$.
Then, we show that for each $q \in F$, either
(1) for each $i \in \{1, \ldots, M_c\}$, there exists $\pi \in \PF_i^{cyc,q}$ such that all finite strings generated by 
each extension of $\pi$ as described in conditions (1)-(a) and (1)-(b) cannot be a substring of any word in $Trace(\DS)$, 
hence, invalidating all accepting cycles starting with $q$, or
(2) for each $i \in \{1, \ldots, M_p\}$, there exists $\pi \in \PF_i^{path,q}$ such that all finite strings generated by 
each extension of $\pi$ as described in conditions (2)-(a) and (2)-(b)  cannot be a substring of any word in $Trace(\DS)$, 
hence, invalidating all paths to the accepting state $q$.

In the next section, we discuss a set of conditions whose satisfaction implies the satisfaction of the conditions in (1) and (2) of Proposition \ref{lem:conds}. The satisfaction of these new conditions can be verified algorithmically; hence, their verification is amenable to automation.

\section{Barrier Certificates for Invalidating Substrings}
\label{sec:bc}
Conditions (1) and (2) of Proposition \ref{lem:conds} require considering finite strings of the form
$\alphabet_0 \alphabet_1$ and $\alphabet_0 \tilde{\alphabet}_0\ldots\tilde{\alphabet}_k \alphabet_1$ where
$k \in \naturals$ and $\alphabet_0, \alphabet_1, \tilde{\alphabet}_0, \ldots, \tilde{\alphabet}_k \in 2^{\AP}$.
Lemma \ref{lem:substring2} and Lemma \ref{lem:substring3} provide
a necessary condition for a trajectory of $\DS$ to have a trace with a substring of the form
$\alphabet_0 \alphabet_1$ and $\alphabet_0 \tilde{\alphabet}_0\ldots\tilde{\alphabet}_k \alphabet_1$, $k \in \naturals$, respectively.

\vspace{-1mm}
\begin{lem}
\label{lem:substring2}
Consider $\Sigma_0, \Sigma_1 \subseteq 2^{\AP}$ and a set $\tilde{\Omega} = \{\alphabet_0 \alphabet_1 \ | \ \alphabet_0 \in \Sigma_0, \alphabet_1 \in \Sigma_1\}$ of finite strings.
Let $\Y_0 = \bigcup_{a \in \Sigma_0} \llbracket a \rrbracket$
and $\Y_1 = \bigcup_{a \in \Sigma_1} \llbracket a \rrbracket$.
If there exists a trajectory $x$ of $\DS$ such that
some finite string in $\tilde{\Omega}$ is a substring of a trace of $x$,
then there exist $t_1 > t_0 \geq 0$  and $t_0' \in [t_0, t_1]$ such that
$x(t) \in \Y_0$ for all $t \in [t_0, t_0')$, 
$x(t) \in \Y_1$ for all $t \in (t_0', t_1]$
and $x(t_0') \in \Y_0 \cup \Y_1$.
\end{lem}
\begin{proof}
This follows directly from the definition of trace.
\end{proof}

\begin{lem}
\label{lem:substring3}
Consider $\Sigma_0, \Sigma_1, \tilde{\Sigma} \subseteq 2^{\AP}$ and
a set $\tilde{\Omega} = \{\alphabet_0 \tilde{\alphabet}_0 \ldots \tilde{\alphabet}_k \alphabet_1 \ | \ k \in \naturals, \alphabet_0 \in \Sigma_0, \tilde{\alphabet}_0, \ldots, \tilde{\alphabet}_k \in \tilde{\Sigma}, \alphabet_1 \in \Sigma_1\}$ of finite strings.
Let $\Y_0 = \bigcup_{a \in \Sigma_0} \llbracket a \rrbracket$,
$\Y_1 = \bigcup_{a \in \Sigma_1} \llbracket a \rrbracket$ and
$\tilde{\Y} = \bigcup_{a \in \tilde{\Sigma}} \llbracket a \rrbracket$.
If there exists a trajectory $x$ of $\DS$ such that some finite string in $\tilde{\Omega}$ is a substring of a trace of $x$,
then there exists $t_1 > t_0 \geq 0$ such that
$x(t_0) \in \Y_0$,
$x(t_1) \in \Y_1$ and
$x(t) \in \Y$ for all $t \in [t_0, t_1]$
where $\Y = \Y_0 \cup \Y_1 \cup \tilde{Y}$.
\end{lem}
\begin{proof}
Consider a trajectory $x$ of $\DS$ and 
a finite substring $\word = \alphabet_0 \tilde{\alphabet}_0 \ldots \tilde{\alphabet}_k \alphabet_1$
where $k \in \naturals$ and $\alphabet_0 \in \Sigma_0$, $\tilde{\alphabet}_0, \ldots,$ $\tilde{\alphabet}_k \in \tilde{\Sigma}$ and $\alphabet_1 \in \Sigma_1$.
Suppose $\word$ is a substring of a trace of $x$.
Then, from the definition of trace,
we can conclude that there exist $t_1 > t_0 \geq 0$ such that
$x(t_0) \in \llbracket \alphabet_0 \rrbracket$, $x(t_1) \in \llbracket \alphabet_1 \rrbracket$ and 
for all $t \in [t_0, t_1]$,
$x(t) \in \llbracket \alphabet_0 \rrbracket \cup \llbracket \tilde{\alphabet}_0 \rrbracket \cup \ldots \cup \llbracket \tilde{\alphabet}_k \rrbracket \cup \llbracket \alphabet_1 \rrbracket$, i.e.,
$x(t_0) \in \Y_0$, $x(t_1) \in \Y_1$ and
$x(t) \in \Y$ for all $t \in [t_0, t_1]$.
\end{proof}

We now consider conditions (1)-(a) and (2)-(a) of Proposition \ref{lem:conds}, which require
considering a finite string of the form $\alphabet_0\alphabet_1$ where $\alphabet_0, \alphabet_1 \in 2^{\AP}$.
The following lemma provides a sufficient condition, based on checking the emptiness of set intersection,
for validating that such a finite string cannot
be a substring of any word in $Trace(\DS)$. 

\begin{lem}
\label{lem:set-intersection}
Consider $\Sigma_0, \Sigma_1 \subseteq 2^{\AP}$ and a set $\tilde{\Omega} = \{\alphabet_0 \alphabet_1 \ | \ \alphabet_0 \in \Sigma_0, \alphabet_1 \in \Sigma_1\}$ of finite strings.
Let $\Y_0 = \bigcup_{a \in \Sigma_0} \llbracket a \rrbracket$
and $\Y_1 = \bigcup_{a \in \Sigma_1} \llbracket a \rrbracket$.
Suppose $\cl{\Y_0} \cap \cl{\Y_1} = \emptyset$.
Then, {no finite string in $\tilde{\Omega}$ can be a substring of any word in $Trace(\DS)$}. 
\end{lem}
\begin{proof}
Suppose, in order to establish a contradiction, that there exists a trajectory $x$ of $\DS$
such that some $\omega = \alphabet_0 \alphabet_1 \in \tilde{\Omega}$ is a substring of a trace of $x$.
From Lemma \ref{lem:substring2},
there must exist $t_1 > t_0 \geq 0$ and $t_0' \in [t_0, t_1]$ such that
$x(t) \in \Y_0$ for all $t \in [t_0, t_0')$ and 
$x(t) \in \Y_1$ for all $t \in (t_0', t_1]$.
Furthermore, from the continuity of the trajectories of \eqref{eq:sys}, 
$x(t) \in \Y_0$ for all $t \in [t_0, t_0')$ implies that $x(t) \in \cl{\Y_0}$ for all $t \in [t_0, t_0']$.
Similarly, $x(t) \in \Y_1$ for all $t \in (t_0', t_1]$ implies that $x(t) \in \cl{\Y_1}$ for all $t \in [t_0', t_1]$.
As a result, it must be the case that $x(t_0') \in \cl{\Y_0}$ and $x(t_0') \in \cl{\Y_1}$,
and hence $x(t_0') \in \cl{\Y_0} \cap \cl{\Y_1}$, leading to a contradiction.
\end{proof}

Using the notion of barrier certificate \cite{Prajna04safetyVerification,Prajna05PrimalDual,prajna05thesis},
we provide a sufficient condition for checking that conditions (1) and (2) of Proposition \ref{lem:conds} are satisfied.
First, Corollary \ref{cor:barrier-cert2} combines Lemma \ref{lem:barrier-cert} and Lemma \ref{lem:substring2}
to provide a sufficient condition for validating that a finite string of the form $\alphabet_0 \alphabet_1$ 
where $\alphabet_0, \alphabet_1 \in 2^{\AP}$ cannot be a substring of any word in $Trace(\DS)$.

\begin{cor}
\label{cor:barrier-cert2}
Consider $\Sigma_0, \Sigma_1 \subseteq 2^{\AP}$ and a set $\tilde{\Omega} = \{\alphabet_0 \alphabet_1 \ | \ \alphabet_0 \in \Sigma_0, \alphabet_1 \in \Sigma_1\}$ of finite strings.
Let $\Y_0 = \bigcup_{a \in \Sigma_0} \llbracket a \rrbracket$,
$\Y_1 = \bigcup_{a \in \Sigma_1} \llbracket a \rrbracket$ and
$\Y = \Y_0 \cup \Y_1$.
Suppose there exists a differentiable function $B \ | \ \X \to \reals$ satisfying conditions
(\ref{eq: lem:barrier-cert1})-(\ref{eq: lem:barrier-cert3}).
Then, {no finite string in $\tilde{\Omega}$ can be a substring of any word in $Trace(\DS)$}. 
\end{cor}

Finally, the following corollary combines Lemma \ref{lem:barrier-cert} and Lemma \ref{lem:substring3}
to provide a sufficient condition for validating that a finite string of the form $\alphabet_0 \tilde{\alphabet}_0 \ldots \tilde{\alphabet}_k \alphabet_1$ 
where $k \in \naturals$ and $\alphabet_0, \alphabet_1, \tilde{\alphabet}_0, \ldots, \tilde{\alphabet}_k \in 2^{\AP}$ cannot be a substring of any word in $Trace(\DS)$.

\begin{cor}
\label{cor:barrier-cert3}
Consider $\Sigma_0, \Sigma_1, \tilde{\Sigma} \subseteq 2^{\AP}$ and
a set $\tilde{\Omega} = \{\alphabet_0 \tilde{\alphabet}_0 \ldots \tilde{\alphabet}_k \alphabet_1 \ | \ k \in \naturals, \alphabet_0 \in \Sigma_0, \tilde{\alphabet}_0, \ldots, \tilde{\alphabet}_k \in \tilde{\Sigma}, \alphabet_1 \in \Sigma_1\}$ of finite strings.
Let $\Y_0 = \bigcup_{a \in \Sigma_0} \llbracket a \rrbracket$
$\Y_1 = \bigcup_{a \in \Sigma_1} \llbracket a \rrbracket$,
$\tilde{\Y} = \bigcup_{a \in \tilde{\Sigma}} \llbracket a \rrbracket$ and
$\Y = \Y_0 \cup \Y_1 \cup \tilde{\Y}$.
Suppose there exists a differentiable function $B : \X \to \reals$ satisfying conditions
(\ref{eq: lem:barrier-cert1})-(\ref{eq: lem:barrier-cert3}).
Then, {no finite string in $\tilde{\Omega}$ can be a substring of any word in $Trace(\DS)$}. 
\end{cor}


\section{$\LTLX$ Verification Procedure}
\label{sec:procedure}

Based on the results presented in Section \ref{sec:substrings} and Section \ref{sec:bc},
we propose the following procedure for $\LTLX$ verification of dynamical systems.
\begin{enumerate}
\item Compute $\A_{\neg\varphi}$.
\item Compute $\P^{cyc}(q)$ and $\P^{path}(q)$ for each $q \in F$ using Algorithm \ref{alg:dfs}.
\item For each $q \in F$, carry out the following steps.
\begin{enumerate}
\item Generate $\PF^3(c)$ for each $c \in \P^{cyc}(q)$ and $\PF^3(p)$ for each $p \in \P^{path}(q)$.
(From its definition, $\PF^3(\pi)$ can be easily generated for any given finite path $\pi$ in $\G$.)
\item Check whether condition (1) or condition (2) of Proposition \ref{lem:conds} is satisfied.
Conditions (1)-(a) and (2)-(a) can be checked using Lemma \ref{lem:set-intersection} or Corollary \ref{cor:barrier-cert2} whereas
conditions (1)-(b) and (2)-(b) can be checked using using Corollary \ref{cor:barrier-cert3}.
\begin{itemize}
\item If either condition (1) or condition (2) holds, 
continue to process next accepting state $q \in F$ or terminate and report that $\DS$ satisfies $\varphi$ if all $q \in F$ has been processed.
\item Otherwise, terminate and report the failure for determining whether $\DS$ satisfies $\varphi$ using this procedure.
\end{itemize}
\end{enumerate}
\end{enumerate}

Steps 1-3(a) above can be automated.
{For example, off-the-shelf tools such as LTL2BA, SPIN and LBT can be used to compute of $\A_{\neg\varphi}$ in step 1.}
Checking conditions (1)-(a) and (2)-(a) of Proposition \ref{lem:conds} can be automated based on Lemma \ref{lem:set-intersection}
by employing generalizations of the so-called S-procedure \cite{topcuthesis:08} or special cases of the Positivstellensatz \cite{parrilothesis:00,stengle1974nullstellensatz}. Furthermore, if 
the sets $\X, \X_0, \ldots, \X_N$ can be described by polynomial functions,
then verification of the conditions in Corollary \ref{cor:barrier-cert2} and Corollary \ref{cor:barrier-cert3} can be reformulated (potentialy conservatively) as sum-of-squares feasibility problems \cite{parrilothesis:00,lasserre2001global}. Specifically, Lemma \ref{lem:barrier-sos} provides a set of sufficient conditions
for the existence of a barrier certificate $B$ as required by Lemma \ref{lem:barrier-cert} to determine whether 
condition (1) or condition (2) of Proposition \ref{lem:conds} is satisfied.


\begin{lem}
\label{lem:barrier-sos}
Let $\Y, \Y_0, \Y_1 \subseteq \X$.
Assume that $\cl{\Y_0}$ and $\cl{\Y_1}$ can be defined by the inequality
$g_0(x) \geq 0$ and $g_1(x) \geq 0$, respectively, i.e.,
$\cl{\Y_0} = \{x : \reals^n \ | \ g_0(x) \geq 0\}$ and
$\cl{\Y_1} = \{x : \reals^n \ | \ g_1(x) \geq 0\}$.
Additionally, assume that $\cl{\Y}$ can be defined by the inequality $g(x) \geq 0$. Suppose there exist a polynomial $B$, a constant $\epsilon > 0$
and sum-of-squares polynomials $s_0$, $s_1$, $s_2$ and $s_3$ such that the following expressions are sum-of-squares polynomials
\begin{eqnarray}
\label{eq:prop:barrier-sos1}
&&-B(x) - s_0(x) g_0(x),\\
\label{eq:prop:barrier-sos2}
&&B(x) - \epsilon - s_1(x) g_1(x), and \\
\label{eq:prop:barrier-sos3}
&&-\frac{\partial B}{\partial x}(x)f(x) - s_2(x) g(x) + s_3(x) g_1(x).
\end{eqnarray}
Then, $B$ satisfies conditions (\ref{eq: lem:barrier-cert1})-(\ref{eq: lem:barrier-cert3}).
\end{lem}
\begin{proof}
Consider an arbitrary $x \in \Y_0$.
Then, $g_0(x) \geq 0$.
Furthermore, since (\ref{eq:prop:barrier-sos1}) and $s_0(x)$ are sum-of-squares polynomials, we get that $-B(x) - s_0(x) g_0(x) \geq 0$
and $s_0(x) \geq 0$.
Combining this with $g_0(x) \geq 0$, we obtain $B(x) \leq 0$, satisfying (\ref{eq: lem:barrier-cert1}).
Similarly, we can show that (\ref{eq:prop:barrier-sos3}) being a sum-of-squares polynomial
ensures that (\ref{eq: lem:barrier-cert3}) is satisfied.
Finally, consider (\ref{eq:prop:barrier-sos2}) and an arbitrary $x \in \cl{\Y_1}$.
Using the same argument as before, we get $B(x) - \epsilon \geq 0$.
Since $\epsilon > 0$, we obtain $B(x) > 0$, satisfying (\ref{eq: lem:barrier-cert2}).
\end{proof}

Based on Lemma \ref{lem:barrier-sos}, a function $B : \X \to \reals$ satisfying conditions (\ref{eq: lem:barrier-cert1})-(\ref{eq: lem:barrier-cert3}) 
can be automatically computed by solving the sum-of-squares problem in Lemma \ref{lem:barrier-sos}, 
which is convex and can be parsed, using SOSTOOLS \cite{Prajna02introducingsostools} and SOSOPT \cite{sosopt}, into a semidefinite program,
provided that the vector field $f$ is polynomial or rational.
Note that in Lemma \ref{lem:barrier-sos}, we assume that $\cl{\Y}$, $\cl{\Y_0}$ and $\cl{\Y_1}$ can be described by polynomial functions
$g$, $g_0$ and $g_1$, respectively, for the ease of the presentation.
The result, however, can be easily extended to handle the case where each of these sets are described by a set of polynomial functions.
For example, suppose $\cl{\Y_0} = \{x : \reals^n \ | \ g_{0,1}(x) \geq 0, \ldots, g_{0,k}(x) \geq 0\}$ where $k \in \naturals$
and $g_{0,1}, \ldots, g_{0,k}$ are polynomial functions.
Then, we need to find sum-of-squares polynomials $s_{0,1}, \ldots, s_{0,k}$, rather than only $s_0$.
In addition, rather than requiring that \eqref{eq:prop:barrier-sos1} is a sum-of-squares polynomial, we require that
$-B(x) - s_{0,1}(x) g_{0,1}(x) - \ldots - s_{0,k}(x) g_{0,k}(x)$ is a sum-of-squares polynomial.
The case where other sets are described by a set of polynomial functions can be treated in a similar way.


\section{Discussion} 
\label{sec:diss}

\subsection{Sources of Incompleteness}
\label{ssec:conservativeness}
The $\LTLX$ verification procedure developed in the previous sections is sound but not complete, i.e., if it reports that $\DS$ satisfies $\varphi$, then we can correctly conclude that $\DS$ actually satisfies $\varphi$.
However, if it reports failure, then $\DS$ may or may not satisfy $\varphi$.
The incompleteness is due to various sources of conservatism included in the procedure for $\LTLX$ verification of dynamical systems proposed in Section \ref{sec:procedure}.

First, Proposition \ref{lem:conds} provides only a sufficient condition for verifying that for each $\omega \in \Omega$, where
$\Omega$ is as defined in Section \ref{sec:approach},
there exists a substring $\omega'$ of $\omega$ that cannot be a substring of any word in $Trace(\DS)$.
However, such a sufficient condition only considers substrings $\omega'$ that are in a particular form
since it may not be possible to check all the substrings of all $\omega \in \Omega$
due to the possible infiniteness of $\Omega$.
We provide further discussion on this issue in Section \ref{ssec:extensions}.
Another source of conservatism comes from Lemma \ref{lem:set-intersection}, Corollary \ref{cor:barrier-cert2} and Corollary \ref{cor:barrier-cert3},
which only provide sufficient conditions for verifying that 
{no finite string in the particular form considered in Proposition \ref{lem:conds} can be a substring of any word in $Trace(\DS)$}. 
Finally, Lemma \ref{lem:barrier-sos} introduces another source of conservatism
as only a sufficient condition for the existence of a function 
$B : \X \to \reals$ satisfying conditions (\ref{eq: lem:barrier-cert1})-(\ref{eq: lem:barrier-cert3}) is provided. The conservatism due to this final cause may be reduced by searching for polynomial barrier certificates ($B$) and S-procedure multipliers ($s_0,s_1,s_2$ and $s_3$) of higher degrees.

\subsection{Computational Complexity}
\label{ssec:complexity}
Let $\A_{\neg\varphi} = (Q, 2^{\AP}, \delta, Q_0, F)$.
It can be shown \cite{Baier:PMC2008} that the size $|Q|$ is at most $|\neg \varphi| 2^{|\neg \varphi|}$ where $|\neg\varphi|$ is the length (in terms of the number of operations) of $\neg\varphi$. (In practice, the size $|Q|$ is typically much smaller than this upper limit \cite{Klein06}.)
Let $|E^\G|$ represent the number of edges of $\G$.
Note that from the construction of $\G$ as explained in Section \ref{sec:substrings}, $|E^\G| \leq |Q|^2$ and $|E^\G| \leq |\delta|$ where $|\delta|$ is the number of transitions in $\A_{\neg\varphi}$.
In the worst case, for each $q \in F$, the size of $\P^{cyc}(q)$ is $(|Q|-1)^{|E^\G|-1}$ whereas the size of $\P^{path}(q)$ is $|Q_0|(|Q|-1)^{|E^\G|-1}$.
(Roughly, this is because the length of each path in $\P^{cyc}(q)$ and $\P^{path}(q)$ is at most $|E^\G|+1$ since edges cannot be repeated.
In addition, at each state except the last two states in the path, there are $|Q|-1$ possibilities of the next state since consecutive repetitions of states are not allowed.)
As a result, for each $q \in F$, the total of at most $(|E^\G| - 1)(|Q|-1)^{|E^\G|-1}(1 + |Q_0|)$ subpaths of length 3 need to be considered in Step (3)-(b) of the $\LTLX$ verification procedure described in Section \ref{sec:procedure}. Note that each of these subpaths corresponds to a numerical search for a barrier certificate and S-procedure multipliers that satisfy the conditions in Lemma \ref{lem:barrier-sos}. For the largest degree of the polynomials in \eqref{eq:prop:barrier-sos1}-\eqref{eq:prop:barrier-sos3} and the number $n$ of continuous states, the complexity of this search is polynomial in each when the other fixed.

\subsection{Comparison to Approaches Based on Explicit Discretization of Dynamics}
\label{ssec:comparison}
A common approach for verifying dynamical systems (call $\DS$) subject to $\LTLX$ specifications (call $\varphi$) is to explicitly construct a finite state abstraction $\TS$ of $\DS$ \cite{TabuadaP03,asarin2007hybridization}. We now briefly compare our method to such approaches with respect to their (in)completeness, computational cost, and conservatism. 

Except for certain special cases, $\TS$ is typically not equivalent (i.e., bisimilar \cite{alur2000discrete}) to $\DS$, but rather an over-approximation of $\DS$, i.e., it may contain behaviors that do not exist in $\DS$. Once $\TS$ is constructed, a typical model checking procedure can be employed to check whether $\TS$ satisfies a given $\LTLX$ specification \cite{Baier:PMC2008,clarkebook}. Since $\TS$ is an over-approximation of $\DS$, if $\TS$ satisfies $\DS$, then we can conclude that $\DS$ also satisfies $\varphi$. However, unless $\TS$ is equivalent to $\DS$, no conclusion about the correctness of $\DS$ can be made otherwise. Hence, as our approach is not complete, the approaches based on explicit discretization of the dynamics are typically not complete, except for certain simple dynamics that allows $\TS$ to be constructed such that it is equivalent to $\DS$ \cite{HKPV98}.

Barrier certificates can also be utilized in these alternative approaches, particularly in the construction of $\TS$. For example, we can construct $\TS$ with $|2^{\AP}|$ states where each state in $\TS$ captures the states in $\DS$ that satisfy the corresponding atomic propositions. Lemma \ref{lem:barrier-cert} can be applied to remove transitions between states of $\TS$ that cannot exist in $\DS$.
The computational complexity of this procedure may seem to be less than ours.
However, even if computing barrier certificates can be automated based on Lemma \ref{lem:barrier-sos},
in practice, solving the sum-of-squares problem in Lemma \ref{lem:barrier-sos} often requires some human guidance, 
particularly in selecting proper degrees of polynomials.
Since $\TS$ contains $|2^\AP|$ states, $|2^\AP|^2$ sum-of-squares problems need to be checked.
In our approach, $|2^\AP|^2$ transitions also need to be checked in the worst case.
In practice though, the subpaths of length 3 considered in Step (3)-(b) of the $\LTLX$ verification procedure often do not include all the $|2^\AP|^2$ transitions. 
As a result, our approach allows to solve only the sum-of-squares problems that correspond to transitions that need to be checked based on these length 3 subpaths.
In the example presented in Section \ref{ssec:ex1}, we consider the case where $|\AP| = 3$; hence, $|2^\AP| = 8$. Solving this problem using the alternative approaches requires considering 64 transitions whereas we show in Section \ref{ssec:ex1} that only 2 sum-of-squares problems need to be solved using our approach. 

The approaches based on explicit discretization described above possibly lead to more conservative results than our approach because they typically utilize only Corollary \ref{cor:barrier-cert2} whereas both Corollary \ref{cor:barrier-cert2} and Corollary \ref{cor:barrier-cert3} can be applied in our approach. Consider, for example, a simple NBA $\A_{\neg\varphi}$ shown in Figure \ref{fig:simple-aut}. Suppose no barrier certificates (see Lemma \ref{lem:barrier-cert}) can be found for the absence of trajectories starting from $\llbracket \alphabet_0 \rrbracket$ and reaching $\llbracket \alphabet_1 \rrbracket$ without leaving $\llbracket \alphabet_0 \rrbracket \cup \llbracket \alphabet_1 \rrbracket$, trajectories starting from $\llbracket \alphabet_1 \rrbracket$ and reaching $\llbracket \alphabet_2 \rrbracket$ without leaving $\llbracket \alphabet_1 \rrbracket \cup \llbracket \alphabet_2 \rrbracket$ and trajectories starting from $\llbracket \alphabet_2 \rrbracket$ and reaching $\llbracket \alphabet_3 \rrbracket$ without leaving $\llbracket \alphabet_2 \rrbracket \cup \llbracket \alphabet_3 \rrbracket$. In this case, a finite state abstraction of the dynamical system contains the transitions from $\alphabet_0$ to $\alphabet_1$, from $\alphabet_1$ to $\alphabet_2$, from $\alphabet_2$ to $\alphabet_3$ and from $\alphabet_3$ to $\alphabet_3$, leading to the conclusion that the correctness of the system cannot be verified. Further suppose that a barrier certificate can be found for the absence of trajectories starting from $\llbracket \alphabet_0 \rrbracket$ and reaching $\llbracket \alphabet_2 \rrbracket$ without leaving $\llbracket \alphabet_0 \rrbracket \cup \llbracket \alphabet_1 \rrbracket \cup \llbracket \alphabet_2 \rrbracket$.%
\footnote{See, for example, Figure \ref{f:ex1-phase-portrait}. In this case, we can enlarge $\X_1$ such that there are trajectories starting from $\X_3$ and reaching $\X_1$ without leaving $\X_1 \cup \X_3$ and there are trajectories starting from $\X_1$ and reaching $\X_2$ without leaving $\X_1 \cup \X_2$. However, there are no trajectories starting from $\X_3$ and reaching $\X_1$ without leaving $\X_1 \cup \X_3$.}
This information cannot be utilized in the approaches based on explicit discretization of
dynamics. With our approach, Corollary \ref{cor:barrier-cert3} can be used to conclude that the system is actually correct.

The conservatism of the approaches based on explicit discretization is often reduced by refining the state space partition based on the dynamics, resulting in larger abstract finite state systems \cite{Yordanov10}. As a result, these approaches face a combinatorial blow up in the size of the underlying discrete abstractions, commonly known as the state explosion problem.


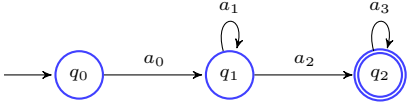
\begin{figure}
\centering
\begin{tikzpicture}[->,>=stealth',shorten >=1pt,auto,node distance=2cm, bend angle=15, font=\scriptsize]
  \tikzstyle{every state}=[circle,thick,draw=blue!75,minimum size=6mm]

   \node [state] (q0) at (0,0) {$q_0$};
   \node [state] (q1) at (2,0){$q_1$};
   \node [state, double] (q2) at (4,0) {$q_2$};
   
   \draw [ shorten >= 1pt, -> ] (-1,0) to (q0);
   
   \path (q0) edge [] node [] {$\alphabet_0$} (q1)
             (q1) edge [loop above] node [] {$\alphabet_1$} (q1)
                     edge [] node [] {$\alphabet_2$} (q2)
             (q2) edge [loop above] node [] {$\alphabet_3$} (q2);
\end{tikzpicture}
\caption{A simple NBA $\A_{\neg\varphi}$ used in the discussion regarding the conservatism of approaches based on explicit discretization of
dynamics compared to our approach.
An arrow without a source points to an initial state. An accepting state is drawn with a double circle.}
\label{fig:simple-aut}
\end{figure}

\subsection{Possible Extensions and Future Work}
\label{ssec:extensions}
Throughout the paper, we consider a continuous vector field to ensure that $x$ is sufficiently smooth,
as required by Lemma \ref{lem:barrier-cert} and Lemma \ref{lem:set-intersection}, partly for ease of presentation. The approach presented in this paper, however, can potentially be extended to handle more general dynamics.
For example, barrier certificates for safety verification of hybrid systems \cite{prajna05thesis} 
can be utilized to extend Lemma \ref{lem:barrier-cert} to handle hybrid systems.
Such certificates, together with additional conditions to handle discrete jumps in Lemma \ref{lem:set-intersection},
allow an extension of our approach to hybrid systems.
Stochastic systems can potentially be handled using a similar idea.
Such an extension is subject to future work.

Based on Proposition \ref{lem:conds}, we only consider subpaths of length 3.
This restriction is due to the property that for any path $\pi$ from $q$ to $q'$,
there exists a path in $\P(q, q')$ 
whose all subpaths of length 3 can be extended in a simple way (by including possibly consecutive state repetitions)
to be subpaths of $\pi$.
However, this property may not necessarily hold for longer subpaths.
For example, consider a graph $\G$ with $V^\G = \{q_0, q_1, q_2, q_3, q_4\}$ and
$E^\G = \{(q_0, q_1), (q_1, q_2), (q_2, q_3),$ $(q_3, q_1), (q_2, q_4)\}$.
In this case, $\P(q_0, q_4) = \{q_0q_1q_2q_4\}$.
Consider a path $\pi = q_0q_1q_2q_3q_1q_2q_4$.
There does not exist any path in $\P(q_0, q_4)$ whose all subpaths of length greater than 3
can be extended only by including possibly consecutive state repetitions to be subpaths of $\pi$.
It is possible to consider longer subpaths, provided that other ways of ``extending'' a subpath
or other finite representative set of paths than those without any repeated edges or consecutive repetitions of states
are considered.
Note also that it is not useful to consider subpaths of length shorter than 3 since invalidating those subpaths requires
proving that no trajectory can reach a certain region, say $\tilde{X}$, no matter where it starts.
Such a condition cannot be verified since a trajectory that starts in $\tilde{X}$ always reaches $\tilde{X}$. 

Including longer subpaths helps reduce the conservatism of our approach.
As the length of subpaths approaches infinity, we recover the set $\Omega$, not only a set of its subpaths.
An example similar to that provided in Section \ref{ssec:comparison} can be constructed to show that considering longer subpaths could help reduce the conservatism of our approach. 
However, including longer subpaths results in increasing computational complexity.

\section{Example}
\label{ssec:ex1}
Consider the problem defined in Example \ref{ex:running-ex}.
As shown in Example \ref{ex:paths},
Algorithm \ref{alg:dfs} yields
$\P^{cyc}(q_4) = \{q_4\}$ and $\P^{path}(q_4) = \{\pi_1, \pi_2, \pi_3\}$ where
\begin{eqnarray*}
&\pi_1 =  q_0q_1q_4, \hspace{3mm}
\pi_2 = q_0q_2q_3q_4, \hspace{3mm}
\pi_3 = q_0q_3q_4.
\end{eqnarray*}


%
%


Since $\P^{cyc}(q_4)$ only contains one path $p = q_4$ and $\PF^3(p) = \emptyset$,
conditions (1) of Proposition \ref{lem:conds} cannot be satisfied.
Hence, we consider condition (2), which requires checking all paths in $\P^{path}(q_4)$.

First, consider $\pi_1 = q_0q_1q_4$.
In this case, we get $\PF^3(\pi_1) = \{\pi_1\}$.
In addition, $\ST(\pi_1) = \{\alphabet_0 \alphabet_1 \ | \ p_0 \in \alphabet_0, p_2 \in \alphabet_1\}$.
Since $\cl{\X_0} \cap \cl{\X_2} = \emptyset$, we can conclude, using Lemma \ref{lem:set-intersection},
that {no finite string in $\ST(\pi_1)$ can be a substring of any word in $Trace(\DS)$}. 
Since $(q_1, q_1) \in E^\G$, we also need to consider all finite strings in
$\ST(q_0 q_1 q_1^+ q_4) = \{\alphabet_0 \tilde{\alphabet}_0 \ldots \tilde{\alphabet}_k \alphabet_1 \ | \ k \in \naturals, p_0 \in \alphabet_0, p_1 \not\in \tilde{\alphabet}_0, \ldots, \tilde{\alphabet}_k, p_2 \in \alphabet_1\}$.
Let $\Y_0 = \X_0$, $\tilde{Y} = \X \setminus \X_1$, $\Y_1 = \X_2$ and $\Y = \Y_0 \cup \Y_1 \cup \tilde{\Y} = \X \setminus \X_1$.
Using SOSOPT, a polynomial $B$ of degree 10, a constant $\epsilon > 0$ and the corresponding sum-of-squares polynomials
$s_0(x), \ldots, s_3(x)$ that make (\ref{eq:prop:barrier-sos1})-(\ref{eq:prop:barrier-sos3}) sum-of-squares polynomials can be computed.
Thus, we can conclude, using Corollary \ref{cor:barrier-cert3}, that
{no finite string in $\ST(q_0 q_1 q_1^+ q_4)$ can be a substring of any word in $Trace(\DS)$}. 
The zero level sets of $B$ and $\frac{\partial B}{\partial x}(x)f(x)$ are depicted in Figure \ref{f:ex1-path1},
showing that $B(x) \leq 0$ for all $x \in \X_0$, $B(x) > 0$ for all $x \in \X_2$ and
$\frac{\partial B}{\partial x}(x)f(x) \leq 0$ for all $x \in (\cl{\X \setminus \X_1}) \setminus \X_2$.

\begin{figure}[t]
\centering
\includegraphics[width=0.35\textwidth]{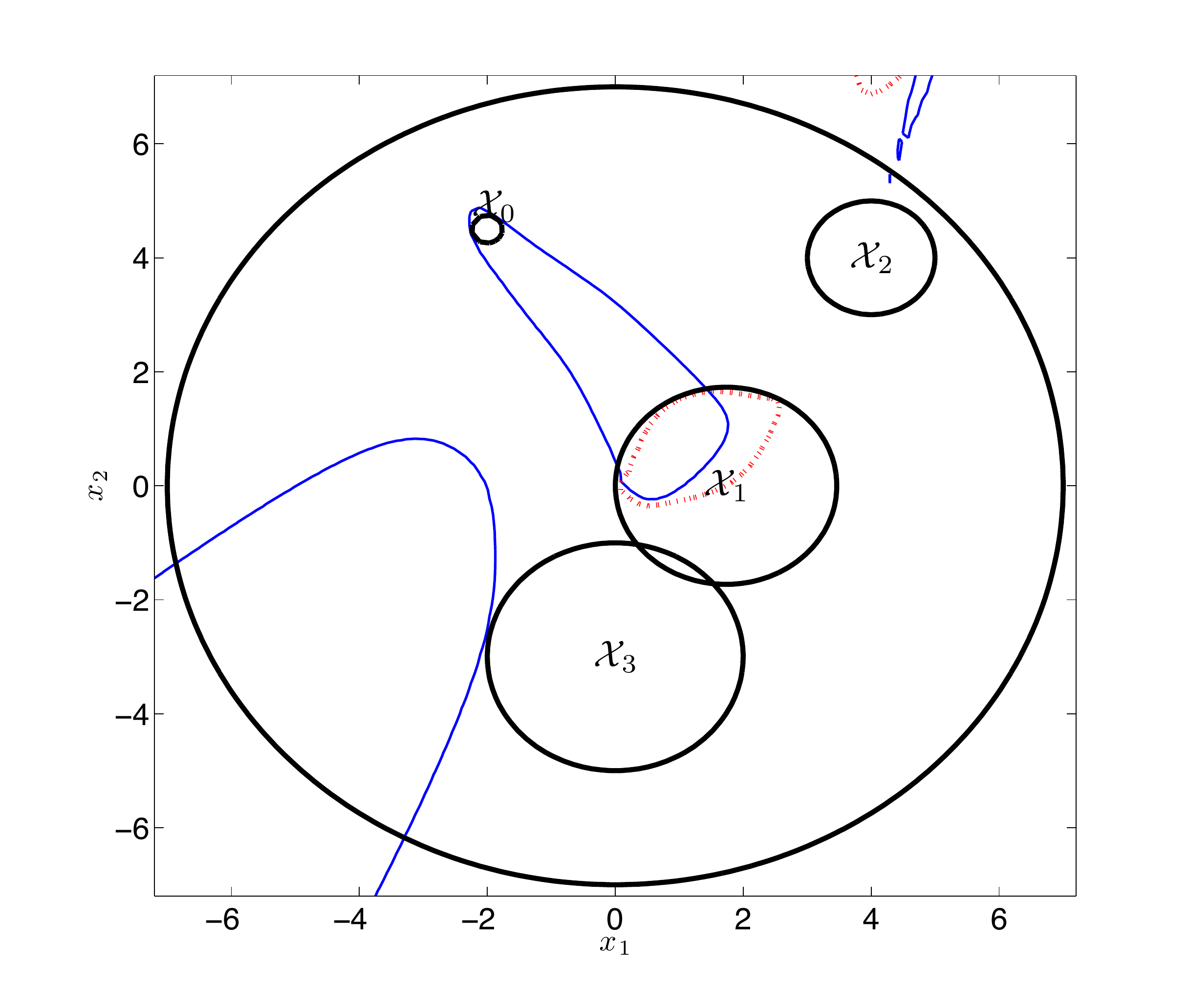}
\caption{The zero level sets of $B$ (light solid blue curves) and $\frac{\partial B}{\partial x}(x)f(x)$ (dotted red curves) for $\pi_1$ where
$\Y_0 = \X_0$, $\Y_1 = \X_2$ and $\Y = \X \setminus \X_1$.}
\label{f:ex1-path1}
\end{figure}

Next, consider $\pi_2 = q_0q_2q_3q_4$.
In this case, $\PF^3(\pi_2) = \{q_0q_2q_3, q_2q_3q_4\}$.
Let $\pi_2' = q_2q_3q_4$.
As for the case of $\pi_1$,  
we can conclude that
{no finite string in $\ST(\pi_2')$ can be a substring of any word in $Trace(\DS)$}
because $\cl{\X_2} \cap \cl{\X_3} = \emptyset$.
Furthermore, for 
$\ST(q_2 q_3 q_3^+ q_4) = \{\alphabet_0 \tilde{\alphabet}_0 \ldots \tilde{\alphabet}_k \alphabet_1 \ | \ k \in \naturals, p_2 \in \alphabet_0, p_3 \in \alphabet_1\}$, 
we let $\Y_0 = \X_2$, $\tilde{Y} = \X$, $\Y_1 = \X_3$ and $\Y = \Y_0 \cup \Y_1 \cup \tilde{\Y} = \X$.
SOSOPT generates a polynomial $B$ of degree 8, a constant $\epsilon > 0$ and the corresponding sum-of-squares polynomials
$s_0, \ldots, s_3$ that make (\ref{eq:prop:barrier-sos1})-(\ref{eq:prop:barrier-sos3}) sum-of-squares polynomials,
ensuring that any trajectory of (\ref{ex1:dyn}) that starts in $\X_2$ cannot reach $\X_3$ without leaving $\cl{\X}$.
Thus, we can conclude, using Corollary \ref{cor:barrier-cert3}, that
{no finite string in $\ST(q_2 q_3 q_3^+ q_4)$ can be a substring of any word in $Trace(\DS)$}. 
The zero level set of $B$ is depicted in Figure \ref{f:ex1-path2},
showing that $B(x) \leq 0$ for all $x \in \X_2$, $B(x) > 0$ for all $x \in \X_3$.
Since $\frac{\partial B}{\partial x}(x)f(x) < 0$ for all $x \in \X$, the zero level set of $\frac{\partial B}{\partial x}(x)f(x)$
is not shown.

\begin{figure}[t]
\centering
\includegraphics[width=0.35\textwidth]{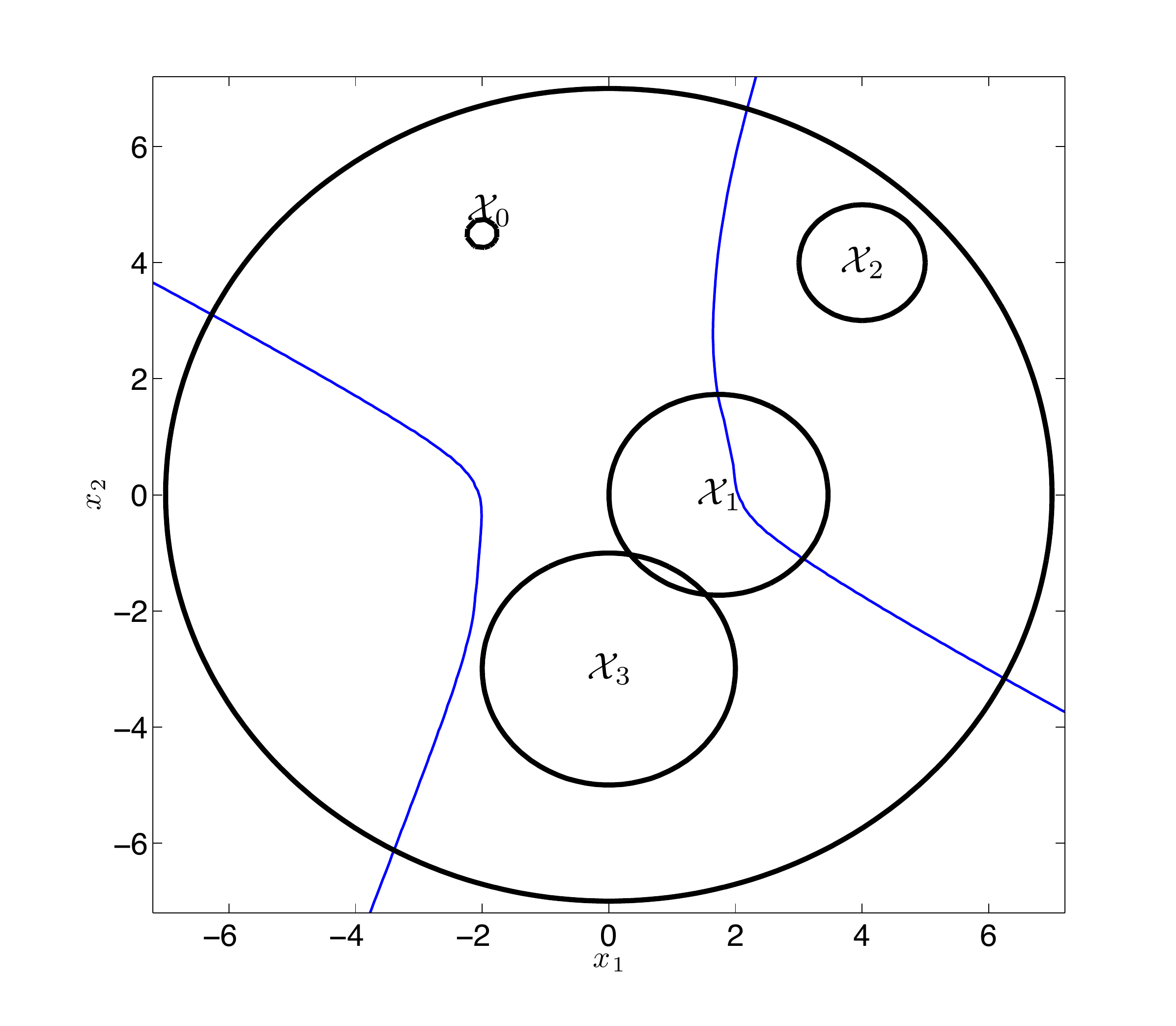}
\caption{The zero level set of $B$ (light solid blue curves) for $\pi_2'$ where
$\Y_0 = \X_2$, $\Y_1 = \X_3$ and $\Y = \X$.}
\label{f:ex1-path2}
\end{figure}

Finally, consider $\pi_3 = q_0 q_3 q_4$.
In this case, $\PF^3(\pi_3) = \{\pi_3\}$.
Furthermore, $\ST(\pi_3) = \ST(\pi_2')$ and $\ST(q_0 q_3 q_3^+ q_4) = \ST(q_2 q_3 q_3^+ q_4)$.
Thus, we can use the results from $\pi_2'$ to conclude that
{no finite string in $\ST(\pi_3) \cup \ST(q_0 q_3 q_3^+ q_4)$ can be a substring of any word in $Trace(\DS)$}. 

At this point, we have checked all the paths in $\P^{path}(q_4)$ to conclude that condition (2) of Proposition \ref{lem:conds} is satisfied.
Thus, we can conclude that $\DS$ satisfies $\varphi$.

\section{Conclusions}
An approach for computational verification of (possibly nonlinear) dynamical systems evolving over continuous state spaces
subject to temporal logic specifications is presented.
Typically, such verification requires checking the emptiness of the intersection of two sets,
the set of all the possible behaviors of the system and the set of all the possible incorrect behaviors,
both of which are potentially infinite, making the verification task challenging (if not impractical).
In order to deal with these infinite sets, we propose a set of strings that,
based on automata theory, can be used to represent the set of all the possible incorrect behaviors.
Our approach then relies on constructing barrier certificates to ensure that each string in this set
cannot be generated by any trajectory of the system.
This integration of automata-based verification and barrier certificates allows us to 
avoid computing an explicit finite state abstraction of the continuous state space based on the underlying dynamics
as commonly done in literature.
Future work includes extending the presented approach to handle more general dynamics 
and attacking various sources of conservatism as discussed in the paper.

\section*{Acknowledgments}
This work was supported in part by the AFOSR (FA9550-12-1-0302) and
ONR (N00014-13-1-0778).
The authors gratefully acknowledge Richard Murray for inspiring discussions.


\bibliographystyle{IEEEtran}
\bibliography{sos-ltl-V3-full}

\begin{thebibliography}{10}
\providecommand{\url}[1]{#1}
\csname url@samestyle\endcsname
\providecommand{\newblock}{\relax}
\providecommand{\bibinfo}[2]{#2}
\providecommand{\BIBentrySTDinterwordspacing}{\spaceskip=0pt\relax}
\providecommand{\BIBentryALTinterwordstretchfactor}{4}
\providecommand{\BIBentryALTinterwordspacing}{\spaceskip=\fontdimen2\font plus
\BIBentryALTinterwordstretchfactor\fontdimen3\font minus
  \fontdimen4\font\relax}
\providecommand{\BIBforeignlanguage}[2]{{%
\expandafter\ifx\csname l@#1\endcsname\relax
\typeout{** WARNING: IEEEtran.bst: No hyphenation pattern has been}%
\typeout{** loaded for the language `#1'. Using the pattern for}%
\typeout{** the default language instead.}%
\else
\language=\csname l@#1\endcsname
\fi
#2}}
\providecommand{\BIBdecl}{\relax}
\BIBdecl

\bibitem{Baier:PMC2008}
C.~Baier and J.-P. Katoen, \emph{Principles of Model Checking (Representation
  and Mind Series)}.\hskip 1em plus 0.5em minus 0.4em\relax The MIT Press,
  2008.

\bibitem{clarkebook}
E.~M. Clarke, O.~Grumberg, and D.~A. Peled, \emph{Model Checking}.\hskip 1em
  plus 0.5em minus 0.4em\relax MIT Press, 1999.

\bibitem{TabuadaP03}
P.~Tabuada and G.~J. Pappas, ``Model checking {LTL} over controllable linear
  systems is decidable,'' in \emph{Hybrid Systems: Computation and Control},
  2003, pp. 498--513.

\bibitem{asarin2007hybridization}
E.~Asarin, T.~Dang, and A.~Girard, ``Hybridization methods for the analysis of
  nonlinear systems,'' \emph{Acta Informatica}, vol.~43, no.~7, pp. 451--476,
  2007.

\bibitem{alur2000discrete}
R.~Alur, T.~A. Henzinger, G.~Lafferriere, and G.~J. Pappas, ``Discrete
  abstractions of hybrid systems,'' \emph{Proceedings of the IEEE}, vol.~88,
  no.~7, pp. 971--984, 2000.

\bibitem{girard2007approximation}
A.~Girard and G.~J. Pappas, ``Approximation metrics for discrete and continuous
  systems,'' \emph{IEEE Transactions on Automatic Control}, vol.~52, no.~5, pp.
  782--798, 2007.

\bibitem{HKPV98}
T.~A. Henzinger, P.~W. Kopke, A.~Puri, and P.~Varaiya, ``What's decidable about
  hybrid automata?'' \emph{Journal of Computer and System Sciences}, vol.~57,
  pp. 94--124, 1998.

\bibitem{Vinter:1980}
R.~Vinter, ``A characterization of the reachable set for nonlinear control
  systems,'' \emph{SIAM Journal on Control and Optimization}, vol.~18, no.~6,
  pp. 599--610, 1980.

\bibitem{prajna05thesis}
S.~Prajna, ``Optimization-based methods for nonlinear and hybrid systems
  verification,'' {Ph.D.} Dissertation, California Institute of Technology,
  2005.

\bibitem{parrilothesis:00}
P.~Parrilo, ``Structured semidefinite programs and semialgebraic geometry
  methods in robustness and optimization,'' {Ph.D.} Dissertation, California
  Institute of Technology, 2000, available at {\tt
  http://thesis.library.caltech.edu/1647/}.

\bibitem{topcuthesis:08}
U.~Topcu, ``Quantitative local analysis of nonlinear systems,'' {Ph.D.}
  Dissertation, UC, Berkeley, July 2008, available at
  {\texttt{http://jagger.me.berkeley.edu/\~{}utopcu/dissertation}}.

\bibitem{prajna2004nonlinear}
S.~Prajna, A.~Papachristodoulou, and F.~Wu, ``Nonlinear control synthesis by
  sum of squares optimization: A lyapunov-based approach,'' in \emph{Asian
  Control Conference}, 2004, pp. 157--165.

\bibitem{jarvis2003some}
Z.~Jarvis-Wloszek, R.~Feeley, W.~Tan, K.~Sun, and A.~Packard, ``Some controls
  applications of sum of squares programming,'' in \emph{Conference Decision
  and Control}, vol.~5, 2003, pp. 4676--4681.

\bibitem{simpaperTopcu}
U.~Topcu, A.~Packard, and P.~Seiler, ``Local stability analysis using
  simulations and sum-of-squares programming,'' \emph{Automatica}, vol.~44, pp.
  2669 -- 2675, 2008.

\bibitem{Tedrake:2010:LFM:1825484.1825491}
R.~Tedrake, I.~R. Manchester, M.~Tobenkin, and J.~W. Roberts, ``{LQR}-trees:
  {Feedback} motion planning via sums-of-squares verification,'' \emph{Int. J.
  Rob. Res.}, vol.~29, no.~8, pp. 1038--1052, 2010.

\bibitem{Galton87TemporalLogics}
A.~Galton, Ed., \emph{Temporal Logics and Their Applications}.\hskip 1em plus
  0.5em minus 0.4em\relax San Diego, CA: Academic Press Professional, Inc.,
  1987.

\bibitem{Kloetzer08}
M.~Kloetzer and C.~Belta, ``A fully automated framework for control of linear
  systems from temporal logic specifications,'' \emph{IEEE Transactions on
  Automatic Control}, vol.~53, no.~1, pp. 287--297, 2008.

\bibitem{LOTM13}
J.~Liu, N.~Ozay, U.~Topcu, and R.~M. Murray, ``Synthesis of reactive switching
  protocols from temporal logic specifications,'' \emph{IEEE Transactions on
  Automatic Control}, vol.~58, no.~7, pp. 1771--1785, 2013.

\bibitem{he08-ltltobuchi}
A.~He, J.~Wu, and L.~Li, ``An efficient algorithm for transforming {LTL}
  formula to {B}\"{u}chi automaton,'' \emph{Intelligent Computation Technology
  and Automation, International Conference on}, vol.~1, pp. 1215--1219, 2008.

\bibitem{gastin01-fastltl}
P.~Gastin and D.~Oddoux, ``Fast {LTL} to {B}\"{u}chi automata translation,'' in
  \emph{CAV '01: Proceedings of the 13th International Conference on Computer
  Aided Verification}.\hskip 1em plus 0.5em minus 0.4em\relax London, UK:
  Springer-Verlag, 2001, pp. 53--65.

\bibitem{ltl2ba}
D.~Oddoux and P.~Gastin, ``{LTL2BA} : fast translation from {LTL} formulae to
  {B}\"{u}chi automata, version 0.2.2 beta,''
  \url{http://www.lsv.ens-cachan.fr/$\sim$gastin/ltl2ba/}.

\bibitem{spin}
G.~J. Holzmann, ``{SPIN} model checker,'' \url{http://spinroot.com/spin/}.

\bibitem{lbt}
M.~R\"{o}nkk\"{o}, H.~Tauriainen, and M.~M\"{a}kel\"{a}, ``{LBT}: {LTL} to
  {B}\"{u}chi conversion,''
  \url{http://www.tcs.hut.fi/Software/maria/tools/lbt/}.

\bibitem{Khalil:Nonlinear96}
H.~K. Khalil, \emph{Nonlinear Systems}.\hskip 1em plus 0.5em minus 0.4em\relax
  Prentice-Hall, 1996.

\bibitem{Russell:AI1995}
S.~J. Russell and P.~Norvig, \emph{Artificial Intelligence: A Modern
  Approach}.\hskip 1em plus 0.5em minus 0.4em\relax Prentice Hall, 1995.

\bibitem{Prajna04safetyVerification}
S.~Prajna and A.~Jadbabaie, ``Safety verification of hybrid systems using
  barrier certificates,'' in \emph{Hybrid Systems: Computation and Control},
  ser. LNCS, R.~Alur and G.~J. Pappas, Eds., vol. 2993.\hskip 1em plus 0.5em
  minus 0.4em\relax Springer, 2004, pp. 477--492.

\bibitem{Prajna05PrimalDual}
S.~Prajna and A.~Rantzer, ``Primal-dual tests for safety and reachability,'' in
  \emph{Hybrid Systems: Computation and Control}, ser. LNCS, M.~Morari and
  L.~Thiele, Eds., vol. 3414.\hskip 1em plus 0.5em minus 0.4em\relax Springer,
  2005, pp. 542--556.

\bibitem{stengle1974nullstellensatz}
G.~Stengle, ``A nullstellensatz and a positivstellensatz in semialgebraic
  geometry,'' \emph{Mathematische Annalen}, vol. 207, no.~2, pp. 87--97, 1974.

\bibitem{lasserre2001global}
J.~B. Lasserre, ``Global optimization with polynomials and the problem of
  moments,'' \emph{SIAM Journal on Optimization}, vol.~11, no.~3, pp. 796--817,
  2001.

\bibitem{Prajna02introducingsostools}
S.~Prajna, A.~Papachristodoulou, and P.~A. Parrilo, ``Introducing {SOSTOOLS}: A
  general purpose sum of squares programming solver,'' in \emph{Proceedings of
  the 41st IEEE Conf. on Decision and Control}, 2002, pp. 741--746.

\bibitem{sosopt}
P.~Seiler, ``{SOSOPT}: A toolbox for polynomial optimization,'' 2013,
  arXiv:1308.1889.

\bibitem{Klein06}
J.~Klein and C.~Baier, ``Experiments with deterministic $\omega$-automata for
  formulas of linear temporal logic,'' \emph{Theoretical Computer Science},
  vol. 363, no.~2, pp. 182--195, 2006.

\bibitem{Yordanov10}
B.~Yordanov and C.~Belta, ``Formal analysis of discrete-time piecewise affine
  systems,'' \emph{IEEE Transactions on Automatic Control}, vol.~55, no.~12,
  pp. 2834--2840, 2010.

\end{thebibliography}

\end{document}